\documentclass[a4paper, 11pt]{article}
\usepackage{amsmath,amssymb, amsthm, amsfonts}
\usepackage{hyperref}
\usepackage{srcltx}
\usepackage{graphicx}
\usepackage[hang]{subfigure}
\usepackage{overpic}
\usepackage{color}
\usepackage{ulem}
\allowdisplaybreaks[4]

\newtheorem{theorem}{Theorem}
\newtheorem{lemma}{Lemma}
\newtheorem{definition}{Definition}

{\theoremstyle{remark} }

\newcommand\dint{\displaystyle\int}
\newcommand\constantsymbol{{\frac{m}{\hat{h}^3}}}
\newcommand\Temperature{RT}
\newcommand\zz{{\mathfrak{z}}}
\newcommand\Li{\mathrm{Li}}
\newcommand\Lihalf[1]{{\Li_{\frac{#1}{2}}}}
\newcommand\weight{\omega}
\newcommand\weightut[3]{{\weight^{[#1,#2,#3]}}}
\newcommand\weightzut{\weightut{\zz}{\bu}{\Temperature}}
\newcommand\mG{\mathcal{G}}

\newcommand\bbH{\mathbb{H}}
\newcommand\bbHut[3]{{\bbH^{[#1,#2,#3]}}}
\newcommand\bbHzut{\bbHut{\zz}{\bu}{\Temperature}}
\newcommand\mP{\mathcal{P}}
\newcommand\identity{{\bf I}}

\newcommand\bbR{\mathbb{R}}

\newcommand\bx{\boldsymbol{x}}

\newcommand\bc{\boldsymbol{c}}
\newcommand\bu{\boldsymbol{u}}
\newcommand\bv{\boldsymbol{v}}

\newcommand\bA{{\bf{A}}}
\newcommand\bB{{\bf{B}}}
\newcommand\bD{{\bf{D}}}
\newcommand\bM{{\bf{M}}}
\newcommand\dd{\,\mathrm{d}}

\newcommand\bw{\boldsymbol{w}}

\newcommand\pd[2]{\dfrac{\partial {#1}}{\partial {#2}}}
\newcommand\opd[2]{\dfrac{\dd {#1}}{\dd {#2}}}
\newcommand\rmspan{\mathrm{span}}

\newcommand\comment[1]{}

\newcommand\revised[2]{#2} 
\newcommand\delete[1]{\revised{#1}{}}
\newcommand\add[1]{\revised{}{#1}}

\graphicspath{{images/}}

\title{13-Moment System with Global Hyperbolicity\\
 for Quantum Gas}

\author{Yana Di\thanks{LSEC, Institute of Computational
    Mathematics and Scientific/Engineering Computing,
    NCMIS, AMSS, Chinese Academy of Sciences, Beijing, China, email: {\tt yndi@lsec.cc.ac.cn}}, 
    ~~Yuwei Fan\thanks{School
    of Mathematical Sciences, Peking University, Beijing, China,
    email: {\tt ywfan@pku.edu.cn}.},
    ~~and Ruo Li\thanks{HEDPS \& CAPT,
    LMAM \& School of Mathematical Sciences, Peking University,
    Beijing, China, email: {\tt rli@math.pku.edu.cn}.}  }

\begin{document}
\maketitle
\begin{abstract}
We point out that the quantum Grad's 13-moment system
\cite{QuantumGrad13} is lack of global hyperbolicity, and even
worse, the thermodynamic equilibrium is not an interior point of
the hyperbolicity region of the system. To remedy this problem, 
\revised{we follow the new theory developed in
    \mbox{\cite{framework,Fan2015}} to propose a regularization for the
system.}{by fully considering Grad's expansion, we split the expansion
into the equilibrium part and the non-equilibrium part, and propose a
regularization for the system with the help of the new theory
developed in \cite{framework, Fan2015}.}
This provides us a new model which is hyperbolic for all admissible
thermodynamic states, and meanwhile preserves the approximate accuracy
of the original system. It should be noted that this procedure is not
a trivial application of the theory in \cite{framework, Fan2015}.

\vspace*{4mm}
\end{abstract}

\section{Introduction}
The behavior of dilute quantum gas can be modeled by the quantum
Boltzmann equation(QBE), which is also known as Uehling-Uhlenbeck
equation\cite{uehling1933}. Studies on the quantum kinetic equation may
promote the understanding of sorts of quantum effects, for example,
Bose-Einstein condensation(BEC)\cite{BEC1999}. However, the
development of quantum kinetic theory lags far behind the gas kinetic
theory. Due to the complexity of the distribution function, it is quite
different to derive quantum hydrodynamic equations from the QBE. The
progress in the study of gas kinetic theory should be illustrative to
the study of quantum kinetic theory. Recently, Yano has
extended the well-known Grad's 13-moment system in gas kinetic
theory to the quantum case to obtain the quantum Grad's 13-moment
system\cite{QuantumGrad13}, which is a new development and
application of Grad's moment method in quantum kinetic theory.

As is well known in gas kinetic theory, there are a number of defects
of Grad's 13-moment system, one of which is that Grad's 13-moment
system are not globally hyperbolic. Precisely, for 1D flow, the
hyperbolicity can only be obtained near the thermodynamic
equilibrium\cite{Muller}. And for 3D case, there is no neighborhood of
thermodynamic equilibrium contained in the hyperbolicity
region\cite{Grad13toR13}. In this paper, we investigate the
hyperbolicity of the quantum Grad's 13-moment system.  By studying the
Jacobian of the convection part of the quantum Grad's 13-moment
system, we find that the quantum Grad's 13-moment system shares the same
defect as the classical Grad's 13-moment system. For 1D flow, the
hyperbolicity of the quantum Grad's 13-moment system can always be
obtained near the equilibrium, and the hyperbolicity region depends on
the equilibrium state \add{due to the quantum effects}, while for the
classical Grad's 13-moment system, the hyperbolicity region only
depends on the non-equilibrium part. For 3D case, \revised{no such
neighborhood of the equilibrium exists in the hyperbolicity
region}{the system is not hyperbolic in any neighborhood of the
equilibrium; that is to say, the equilibrium is on the boundary of the
hyperbolicity region}.

The loss of hyperbolicity directly breaks the well-posedness of the
partial differential equations, thus it is hard to apply the quantum
Grad's 13-moment model to practical problems. In gas kinetic theory,
\add{by investigating the coefficient matrix of the moment system,} a
globally hyperbolic regularization has been proposed to settle the
loss of hyperbolicity of Grad's moment system \cite{Fan, Fan_new}.
\revised{And the regularization has been extended to a framework on
  moment model reduction from kinetic equations to globally hyperbolic
  hydrodynamic system \mbox{in \cite{framework, Fan2015}}}{In
  \cite{framework, Fan2015}, it was revealed that the essential of the
  regularization is to treat the time and space derivatives in the
  same way, then the authors extended the regularization to a
  framework on moment model reduction from kinetic equations to
  globally hyperbolic hydrodynamic system}.  Using the theory therein,
one of the tasks in this paper is to deduce a globally hyperbolic
hydrodynamic model from the quantum Boltzmann equation. Unfortunately,
if one applies the framework \cite{framework, Fan2015} in a trivial
way, the resulting system would be not compatible with Grad's
13-moment system \delete{. This is} due to the particular formation of
the quantum thermodynamic equilibrium.  \add{Precisely, the linearized
  equations of the resulting system at the equilibrium is different
  from those of Grad's 13-moment system, and the first step of
  Maxwellian iteration\cite{Ikenberry} of the resulting system fails
  to give the correct NSF law.}

\revised{As one way out,}{Since the equilibrium plays an essential
role in the quantum Grad's expansion,} we split the expansion into two
parts: the equilibrium part and the non-equilibrium part. 
\delete{and propose a new regularization to obtain the
regularized 13-moment system. By applying the technique in
\mbox{\cite{Grad13toR13}}, we prove that the new moment system is
hyperbolic for any admissible state.}
\add{By applying the framework \cite{framework,Fan2015} only on the
non-equilibrium part, we obtained a new regularized 13-moment system,
which is proved to be hyperbolic for any admissible state.}
\add{The linearized equations of the regularized 13-moment model are
the same as those of the quantum Grad's 13-moment system. Hence,}
the regularized 13-moment model differs from the quantum Grad's
13-moment model by only some high-order terms, and thus the new system
preserves most of the merits of the original model. As an example of
these merits, we show that the first step of Maxwellian iteration
\cite{Ikenberry} of the new model also yields the correct NSF law.
\add{Moreover, we bring a new observation on the globally
hyperbolic regularization that the key point of
the regularization is to split the convection operator into the
product of the multiplying velocity operator and the space derivative
operator.}

The layout of this paper is as follows. In Sect. \ref{sec:Grad13},
an overview of the quantum Boltzmann equation and the quantum Grad's
13-moment system are given. In Sect. \ref{sec:hyperbolicity}, the
hyperbolicity of the quantum Grad's 13-moment system is studied in
details \add{both for 1D and 3D case}. The hyperbolic regularization for the
quantum Grad's 13-moment system and the regularized 13-moment system
are then proposed and discussed in Sect. \ref{sec:regularization}.
The difference between the regularized system and the original system
is discussed.  Furthermore, the hyperbolicity and the physical
properties of the regularized 13-moment system are also studied. At
last, a conclusion is presented in Sect. \ref{sec:conclusion}.


\section{Grad's 13-Moment System for QBE}
\label{sec:Grad13}
Denoting the distribution function of quantum gas particles by $f(t,\bx,\bv)$, 
which describes the probability density to find a particle at space point
$\bx$ and time $t$ with velocity $\bv$, we have the quantum Boltzmann equation, 
the well-known Uehling-Uhlenbeck equation \cite{uehling1933}
\begin{equation}\label{eq:QBE}
    \pd{f}{t}+\bv\cdot\nabla_{\bx}f=Q(f,f),
\end{equation}
where the collision term is defined as
\begin{equation}
    Q(f,f) = \int_{\bbR^3}\int_{0}^{2\pi}\int_{0}^{\pi}\left[
        (1-\theta f)(1-\theta f_*)f'f_*' -(1-\theta f')(1-\theta f_*')ff_*
        \right] g\sigma\sin\chi\dd\chi\dd\epsilon\dd\bv_*.
\end{equation}
Here $f$, $f'$, $f_*$ and $f_*'$ are the shorthand notations for
$f(t,\bx,\bv)$, $f(t,\bx,\bv')$, $f(t,\bx,\bv_*)$ and
$f(t,\bx,\bv_*')$. $(\bv,\bv_*)$ and $(\bv',\bv_*')$ are the velocities
before and after collision. $\epsilon$ is the scattering angle, $\chi$
is the deflection angle, $g=|\bv-\bv_*|$ and $\sigma$ is the
differential cross section. $\theta=1,0,-1$ correspond to Fermion,
classical gas and Boson, respectively. Although Fermion and Boson are major considerations of the paper, almost all the results are also valid for
classical gases.

Similar to the classical Boltzmann equation, the quantum collision
term conserves mass, momentum and energy, i.e:
\begin{equation}\label{eq:collisioninvariants}
    \int_{\bbR}Q(f,f) \dd\bv = 0,\quad
    \int_{\bbR}Q(f,f) v_i\dd\bv = 0,~i=1,2,3,\quad
    \int_{\bbR}Q(f,f) |\bv|^2\dd\bv = 0,
\end{equation}
where $1$, $v_i(i=1,2,3)$, $|\bv|^2$ are called collision invariants. 
If we define the density $\rho$, the macroscopic velocity $\bu$ and the pressure
$p$ as 
\begin{equation}
        \rho=\constantsymbol\int_{\bbR^3}f\dd\bv,\quad
        \rho \bu=\constantsymbol \int_{\bbR^3}f\bv\dd\bv,\quad
        p =\frac{1}{3}\constantsymbol\int_{\bbR^3}f|\bv-\bu|^2\dd\bv,
\end{equation}
where $\hat{h}=h/m$, and $m$ is the mass of the particle, $h$ is Planck's constant,
and define the pressure tensor $p_{ij}$ and the heat flux $q_i$ as
\begin{equation}
    p_{ij}=\constantsymbol\int_{\bbR^3}fc_ic_j\dd\bv,~i,j=1,2,3,\quad 
    q_i = \frac{1}{2}\constantsymbol\int_{\bbR^3}fc_i|\bc|^2\dd\bv,~i=1,2,3,
\end{equation}
where $c_i=v_i-u_i$, $i=1,2,3$, then 
\begin{subequations}\label{eq:EulerEqs}
    \begin{align}
        \opd{\rho}{t}&+\rho\pd{u_d}{x_d}=0,\\
        \rho\pd{u_i}{t}&+\pd{p_{id}}{x_d}=0,~i=1,2,3,\\
        \opd{p}{t}&+\frac{2}{3}p_{id}\pd{u_i}{x_d}+p\pd{u_d}{x_d}+\frac{2}{3}\pd{q_d}{x_d}=0,
    \end{align}
\end{subequations}
where $\opd{\cdot}{t}=\pd{\cdot}{t}+u_d\pd{\cdot}{x_d}$.
The thermodynamic equilibrium is 
\begin{equation}
    f_{eq}=\frac{1}{\zz^{-1}\exp\left(
    \frac{|\bv-\bu|^2}{2\Temperature}
\right)+\theta},
\end{equation}
where $\zz$ and $\Temperature$ is related to $\rho$ and $p$ as 
\begin{equation}\label{eq:relation_rhopzt}
    \rho = \constantsymbol\sqrt{2\pi\Temperature}^3\Lihalf{3},\quad 
    p=\constantsymbol\sqrt{2\pi\Temperature}^3\Temperature \Lihalf{5},
\end{equation}
and $\Li_s:=-\theta\Li_s(-\theta\zz)$ is the polylogarithm.
For the special case $\theta=0$, let $\Li_s=\zz$.

In this paper, we focus on 13-moment system, where 13 moments are
$\rho, u_i, p_{ij}$ and $q_i$, $i,j=1,2,3$. Defining
\begin{equation}\label{eq:def_qijk_delta}
    \sigma_{ij}=p_{ij}-p\delta_{ij},\quad
    q_{ijk}=\constantsymbol\int_{\bbR^3}fc_ic_jc_k\dd\bv,\quad 
    \Delta_{ij}=\constantsymbol\int_{\bbR^3}fc_ic_j|\bc|^2\dd\bv,
\end{equation}
we have
\begin{subequations}\label{eq:sigmaq}
    \begin{align}
        \opd{\sigma_{ij}}{t}&+\sigma_{ij}\pd{u_d}{x_d}+2p_{d\langle
        j}\pd{u_{i\rangle}}{x_d}
        +\pd{q_{ijd}}{x_d}-\frac{2}{3}\pd{q_d}{x_d}
        =Q_{ij}^{(2)},~i,j=1,2,3,\\
        \opd{q_i}{t}&+\frac{5}{2}p\opd{u_i}{t}+\sigma_{ij}\opd{u_j}{t}
        +2q_{(i}\pd{u_{d)}}{x_d}+q_{ijd}\pd{u_j}{x_d}
        +\frac{1}{2}\pd{\Delta_{id}}{x_d}=Q_i^{(3)},
        ~i=1,2,3,
    \end{align}
\end{subequations}
where $Q^{(2)}_{ij}=\dint_{\bbR^3}c_ic_jQ(f,f)\dd\bv$ and
$Q^{(3)}_i=\dfrac{1}{2}\dint_{\bbR^3}c_i|\bc|^2Q(f,f)\dd\bv$.
Particularly, for the quantum Bhatnagar-Gross-Krook model\cite{QBGK}, 
\[
    Q^{(2)}_{BGK,ij}=-\frac{1}{\tau}\sigma_{ij},\quad
    Q^{(3)}_{BGK,i}=-\frac{1}{\tau}q_i.
\]

\revised{Equations}{Eqs.} \eqref{eq:EulerEqs} and \eqref{eq:sigmaq} are
the 13-moment system, with $q_{ijk}$ and $\Delta_{ij}$ are
undetermined.  In \cite{QuantumGrad13}, Yano extended Grad's
expansion\cite{Grad} into the quantum case to obtain the quantum
Grad's 13-moment system.  Actually, Yano gave two kinds of expansions
\cite{QuantumGrad13}:
\begin{align}
    f_{G13}^{H1} &= f_{eq}\left( 1
    +\frac{\sigma_{ij}}{2p}\frac{\Lihalf{5}}{\Lihalf{7}}\left(\frac{c_ic_j}{\Temperature}-\delta_{ij}\frac{\Lihalf{5}}{\Lihalf{3}}\right)
    +\frac{\mathfrak{B}^{-1}q_i}{5p\Temperature}
    c_i\left( \frac{|\bc|^2}{\Temperature}-5\frac{\Lihalf{7}}{\Lihalf{5}}\right)
     \right),
    \label{eq:Grad13expansion1}\\
    f_{G13}^{H2} &= f_{eq}+f_{eq}(1-\theta f_{eq})\left(
    \frac{\sigma_{ij}}{2p}\left(\frac{c_ic_j}{\Temperature}-\delta_{ij}\frac{|\bc|^2}{3\Temperature}\right) 
    + \frac{\mathfrak{B}_2^{-1}q_i}{5p\Temperature}c_i\left(
    \frac{|\bc|^2}{\Temperature}-5\frac{\Lihalf{5}}{\Lihalf{3}}\right)\right),
    \label{eq:Grad13expansion2}
\end{align}
where
$\mathfrak{B}=7/2(\Lihalf{9}/\Lihalf{5})-5/2(\Lihalf{7}^2/\Lihalf{5}^2)$,
and $\mathfrak{B}_2=7/2(\Lihalf{7}/\Lihalf{5})-5/2(\Lihalf{5}/\Lihalf{3})$.
Both of the expansions degenerate into Grad's expansion\cite{Grad} for classical gases. In \cite{QuantumGrad13}, the expansion \eqref{eq:Grad13expansion1} was used to derive the
quantum Grad's 13-moment system, and the expansion
\eqref{eq:Grad13expansion2} to assist the analysis of
Grad's moment system (essentially for Chapman-Enskog
expansion\cite{Chapman} and collision term\cite{QuantumFP}). However,
for the quantum Boltzmann equation, these two different expansions
would result in different moment systems. 
Substituting the expansion \eqref{eq:Grad13expansion1} or
\eqref{eq:Grad13expansion2} into \eqref{eq:def_qijk_delta} yields the
expression of $q_{ijk}$ and $\Delta_{ij}$ directly. For the expansion
\eqref{eq:Grad13expansion1}, the moment closure is 
\begin{equation}\label{eq:Grad13closure1}
    \begin{aligned}
        q_{ijk} &= \frac{2}{5}\left(
        \delta_{ij}q_k+\delta_{ik}q_j+\delta_{kj}q_i\right),\\
        \Delta_{ij}&=\add{\constantsymbol}\sqrt{2\pi\Temperature}^3\Temperature^2\Lihalf{7}\left(
        5\delta_{ij}+7\frac{\sigma_{ij}}{p}\frac{\Lihalf{5}\Lihalf{9}}{\Lihalf{7}^2}
         \right).
    \end{aligned}
\end{equation}
For the expansion \eqref{eq:Grad13expansion2}, the moment closure is 
\begin{equation}\label{eq:Grad13closure2}
    \begin{aligned}
        q_{ijk} &= \frac{2}{5}\left(
        \delta_{ij}q_k+\delta_{ik}q_j+\delta_{kj}q_i\right),\\
        \Delta_{ij}&=\add{\constantsymbol}\sqrt{2\pi\Temperature}^3\Temperature^2\Lihalf{7}\left(
        5\delta_{ij}+7\frac{\sigma_{ij}}{p} \right).
    \end{aligned}
\end{equation}
Both the system \eqref{eq:EulerEqs} and \delete{the system}
\eqref{eq:sigmaq} with the moment closure \eqref{eq:Grad13closure1} or
\eqref{eq:Grad13closure2} are Grad's 13-moment system for the quantum
Boltzmann equation. 

In this paper, we follow the work in \cite{QuantumGrad13}, and only consider the
case with the moment closure \eqref{eq:Grad13closure1}. All the deductions
below can be extended to the moment system with the moment closure
\eqref{eq:Grad13closure2} without essential difficulties.


\section{Hyperbolicity of Grad's 13-Moment System}
\label{sec:hyperbolicity}
As is well known, loss of hyperbolicity, which leads to the loss of
well-posedness of the model, is a severe drawback of Grad's 13-moment
system\cite{Muller}. In \cite{Grad13toR13}, the authors pointed out
that the thermodynamic equilibrium is not an interior point of the
hyperbolicity region of Grad's 13-moment system. Therefore, any small
perturbation of the equilibrium state may lead to the loss of
hyperbolicity, thus the existence of the solution of Grad's 13-moment
system is hardly achieved. In this section, we investigate the
hyperbolicity of the quantum Grad's 13-moment system with the moment
closure \eqref{eq:Grad13closure1}, and the results can be extended to
the system with the moment closure \eqref{eq:Grad13closure2} with routine
calculations. Let us first define the hyperbolicity as follows.
\begin{definition}[Hyperbolicity]
\label{def:hyperbolicity}
A system of first order quasi-linear partial differential equations
\[
\pd{\bw}{t}+\bA_d(\bw)\pd{\bw}{x_d}=0
\]
is called {\it hyperbolic} in some region $\Omega$ if and only if any
linear combination of $\bA_d(\bw)$ is diagonalizable with real
eigenvalues for all $\bw\in\Omega$.
\end{definition}

\subsection{Hyperbolicity for 1D case}
For 1D case, the quantum Grad's 13-moment system reduces to a smaller
system containing only five equations, which can be obtained by setting
$u_2=u_3=0$, $p_{12}=p_{12}=p_{23}=0$, $q_2=q_3=0$ and
$p_{22}=p_{33}$. Denoting $\bw_5=(\rho, u_1, p_{11}, q_1, p)^T$, one has
\begin{equation}\label{eq:ms_5}
    \pd{\bw_5}{t}+\bA_5\pd{\bw_5}{x_1}=\boldsymbol{Q}_5,
\end{equation}
where $\boldsymbol{Q}_5=(0,0,Q_{11}^{(2)},Q_{1}^{(3)},0)^T$, and 
\begin{equation}
    \bA_5=
    \begin{pmatrix} 
        u_1&\rho&0&0&0\\ 
        0&u_1&\frac{1}{\rho}&0&0\\ 
        0&3p_{11}&u_1&6/5&0 \\ 
        -a_1 &{\frac{16}{5}}q_1& a_2 &u_1&a_3\\ 
        0&p+\frac{2}{3}p_{11}&0&2/3&u_1
    \end {pmatrix},
\end{equation}
\[
    \begin{aligned}
        a_1&=
        \frac{5p\Temperature \mathfrak{b}}{2\rho}
        \left(\frac{7}{2}\frac{\Lihalf{3}^2\Lihalf{7}}{\Lihalf{1}\Lihalf{5}^2}-\frac{5}{2}\frac{\Lihalf{3}}{\Lihalf{1}}\right) 
        + \frac{7\sigma_{11}\Temperature \mathfrak{b}}{2\rho}
        \left(\frac{\Lihalf{3}^2\Lihalf{9}}{\Lihalf{1}\Lihalf{5}\Lihalf{7}}-\frac{5}{2}\left(\frac{\Lihalf{3}}{\Lihalf{1}}-\frac{\Lihalf{3}\Lihalf{5}\Lihalf{9}}{\Lihalf{1}\Lihalf{7}^2}\right)\right),
        \\
        a_2&=\frac{7\Temperature\Lihalf{9}}{2\Lihalf{7}}-\frac{3p}{2\rho}-\frac{p_{11}}{\rho},\\
        a_3&=
        \frac{5\Temperature}{2}\left( (1+\mathfrak{b})\frac{\Lihalf{7}}{\Lihalf{5}}  
        -
        \frac{3\mathfrak{b}}{2}\frac{\Lihalf{3}}{\Lihalf{1}}\left(1-\frac{\Lihalf{3}\Lihalf{7}}{\Lihalf{5}^2}\right) \right)
        \\
        &\qquad+
        \frac{7\Temperature}{2}\left(\left(\frac{\sigma_{11}\mathfrak{b}}{p}-1\right)\frac{\Lihalf{9}}{\Lihalf{7}}
        -
        \frac{3\mathfrak{b}}{2}\frac{\Lihalf{3}}{\Lihalf{1}}\frac{\sigma_{11}}{p}\left(1-\frac{\Lihalf{5}\Lihalf{9}}{\Lihalf{7}^2}\right)
        \right).
    \end{aligned}
\]
Here
$\mathfrak{b}=\frac{2}{5-3\frac{\Lihalf{3}^2}{\Lihalf{1}\Lihalf{5}}}$.
The characteristic polynomial of $\bA_5$ can be directly calculated as
\begin{equation}
    \begin{aligned}
    |\lambda\identity-\bA_5| &=
    \frac{1}{75}(\lambda-u_1) \left(
    75(\lambda-u_1)^{4}-\left(90{a_2}+50a_3+225\frac{p_{11}}{\rho}\right)(\lambda-u_1)^{2}
    \right)\\
    &\qquad\qquad\qquad\qquad\qquad\left.
    +90\left(a_1+a_3\frac{\sigma_{11}}{\rho}\right)+288\frac{q_1}{\rho}(\lambda-u_1)
    \right).
    \end{aligned}
\end{equation}
Introducing the dimensionless quantity, 
$\hat{\lambda}=(\lambda-u_1)/\sqrt{\Temperature}$, 
\add{we} reduces $|\lambda\identity-\bA_5|=0$ into 
\begin{equation}\label{eq:cp_5}
    \hat{\lambda}\left[
    75\hat{\lambda}^{4}\Temperature^2-\left(90{a_2}+50a_3+225\frac{p_{11}}{\rho}\right)\hat{\lambda}^{2}\Temperature
    +90\left(a_1+a_3\frac{\sigma_{11}}{\rho}\right)+288\frac{q_1}{\rho}\hat{\lambda}\sqrt{\Temperature}
    \right]=0.
\end{equation}
For the special case $\sigma_{11}=0$, $q_1=0$, i.e. the thermodynamic
equilibrium state, \eqref{eq:cp_5} reduces into
\begin{equation}\label{eq:cp_5-2}
    \hat{\lambda}\left[
    75\hat{\lambda}^{4}\Temperature^2-\left(90{a_2}+50a_3+225\frac{p}{\rho}\right)\hat{\lambda}^{2}\Temperature
    +90a_1
    \right]=0.
\end{equation}
Since 
\[
    90a_1>0,\quad 90a_2+50a_3+225\frac{p}{\rho}>0,
\]
and it is easy to check 
\[
    \Delta_{G5} := \left(90a_2+50a_3+225\frac{p}{\rho}\right)^2-4\cdot
    75\Temperature^2\cdot 90a_1>0,
\]
Eq. \eqref{eq:cp_5-2} has five distinct zeros, \add{which read 
\begin{equation}
    \begin{aligned}
    \hat{\lambda}_{1,5} &=
    \pm\sqrt{\frac{90a_2+50a_3+225\frac{p}{\rho}+\sqrt{\Delta_{G5}}}{150\Temperature^2}},\quad
    \hat{\lambda}_3 = 0,\\
    \hat{\lambda}_{2,4} &=
    \pm\sqrt{\frac{90a_2+50a_3+225\frac{p}{\rho}-\sqrt{\Delta_{G5}}}{150\Temperature^2}}.
    \end{aligned}
\end{equation}
In this case, $\bA_5$ has five distinct eigenvalues as $\lambda_i =
u_1+\sqrt{\Temperature}\hat{\lambda}_i$, $i=1,2,\cdots,5$.
}
If $\sigma_{11}/p$ and
$q_1 /(p\sqrt{\Temperature})$ are small enough, the zeros of
\eqref{eq:cp_5} are still real and separable. Hence, the matrix
$\bA_5$ has no multiple eigenvalue around the equilibrium, which means
the system \eqref{eq:ms_5} is hyperbolic around the equilibrium.
Moreover, the hyperbolicity of \eqref{eq:ms_5} depends only on
$\zz$, $\sigma_{11}/p$ and $q_1 / (p\sqrt{\Temperature})$. For a given
$\zz$, we can directly plot the hyperbolicity region of
\eqref{eq:ms_5}. Figs. \ref{fig:hrBose} and \ref{fig:hrFermi} give the
hyperbolicity regions of the system \eqref{eq:ms_5} with $\theta=-1$,
$1$ for some given $\zz$, respectively. It is reminded that
$\sigma_{11}/p\in(-1,2)$ due to the fact that the pressure tensor
$p_{ij}$ is positive definite, and $q_1
/(p\sqrt{\Temperature})\in\bbR$.

\begin{figure}
    \subfigure[$\zz=0.1$]{ 
        \begin{overpic}[width=.3\textwidth]{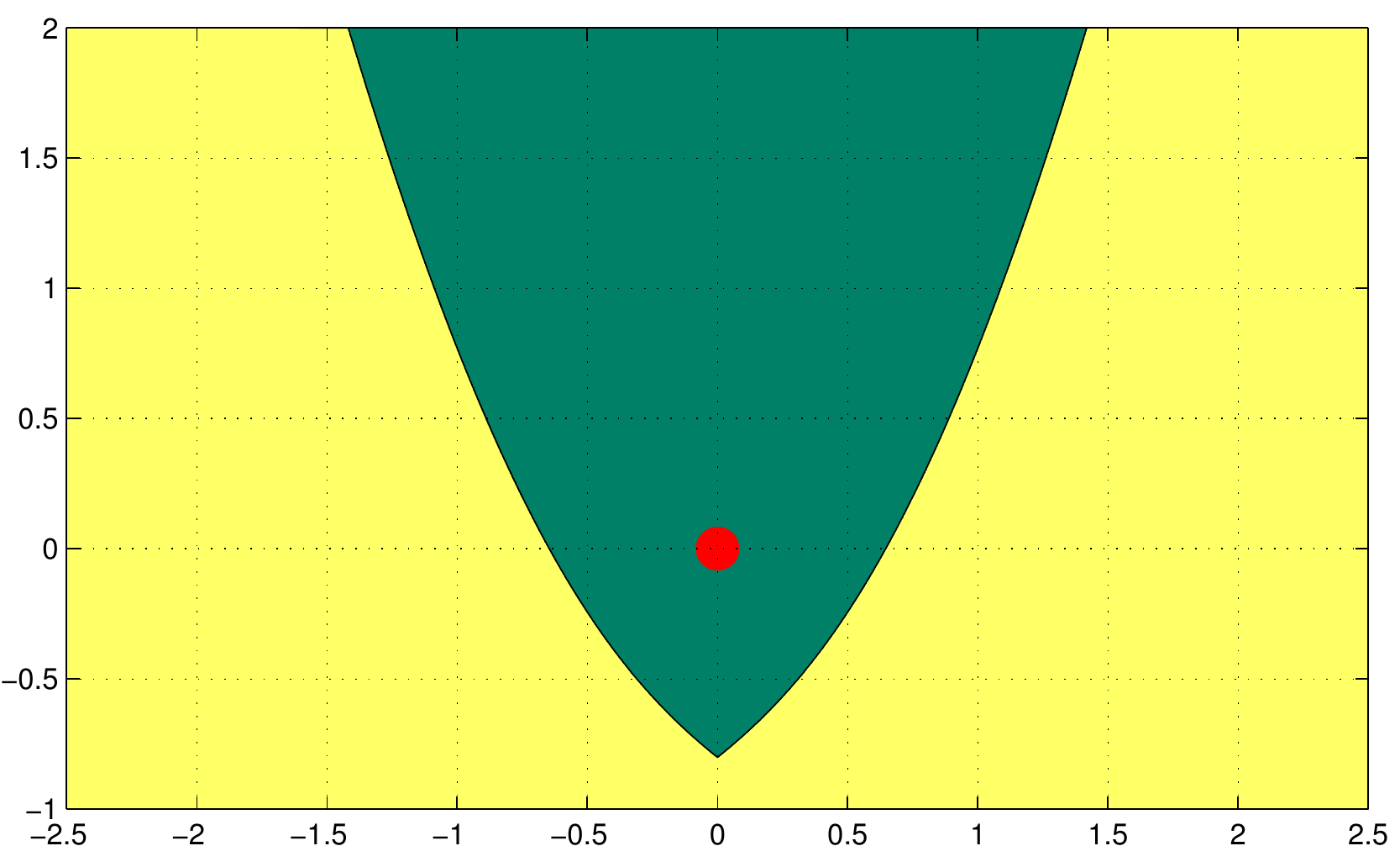}
            \put(76,-4){\tiny$\frac{q_1}{p\sqrt{\Temperature}}$}
            \put(-7,43){\tiny$\frac{\sigma_{11}}{p}$}
        \end{overpic}
    }
    \subfigure[$\zz=0.5$]{ 
        \begin{overpic}[width=.3\textwidth]{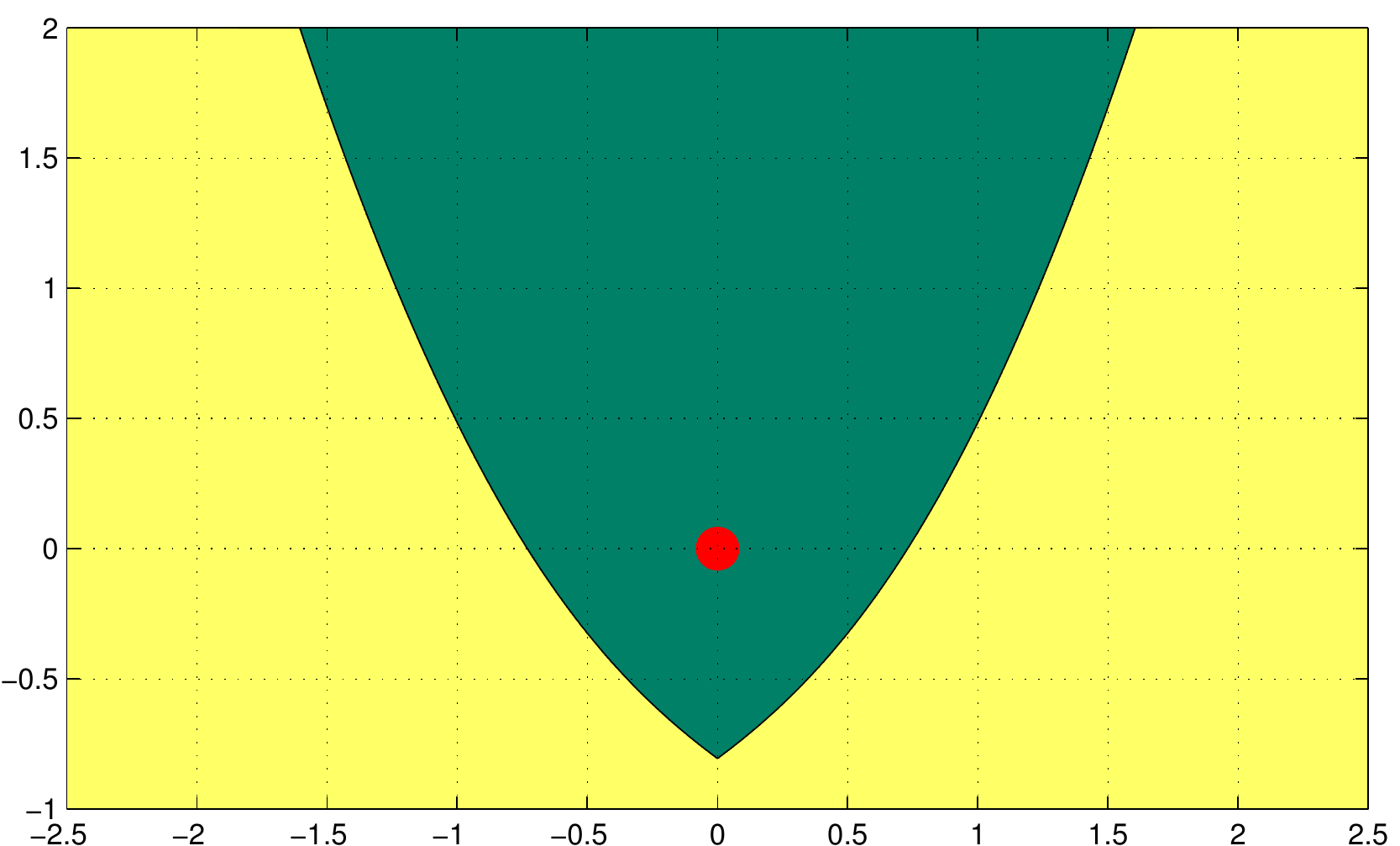}
            \put(76,-4){\tiny$\frac{q_1}{p\sqrt{\Temperature}}$}
            \put(-7,43){\tiny$\frac{\sigma_{11}}{p}$}
        \end{overpic}
    }
    \subfigure[$\zz=0.9$]{ 
        \begin{overpic}[width=.3\textwidth]{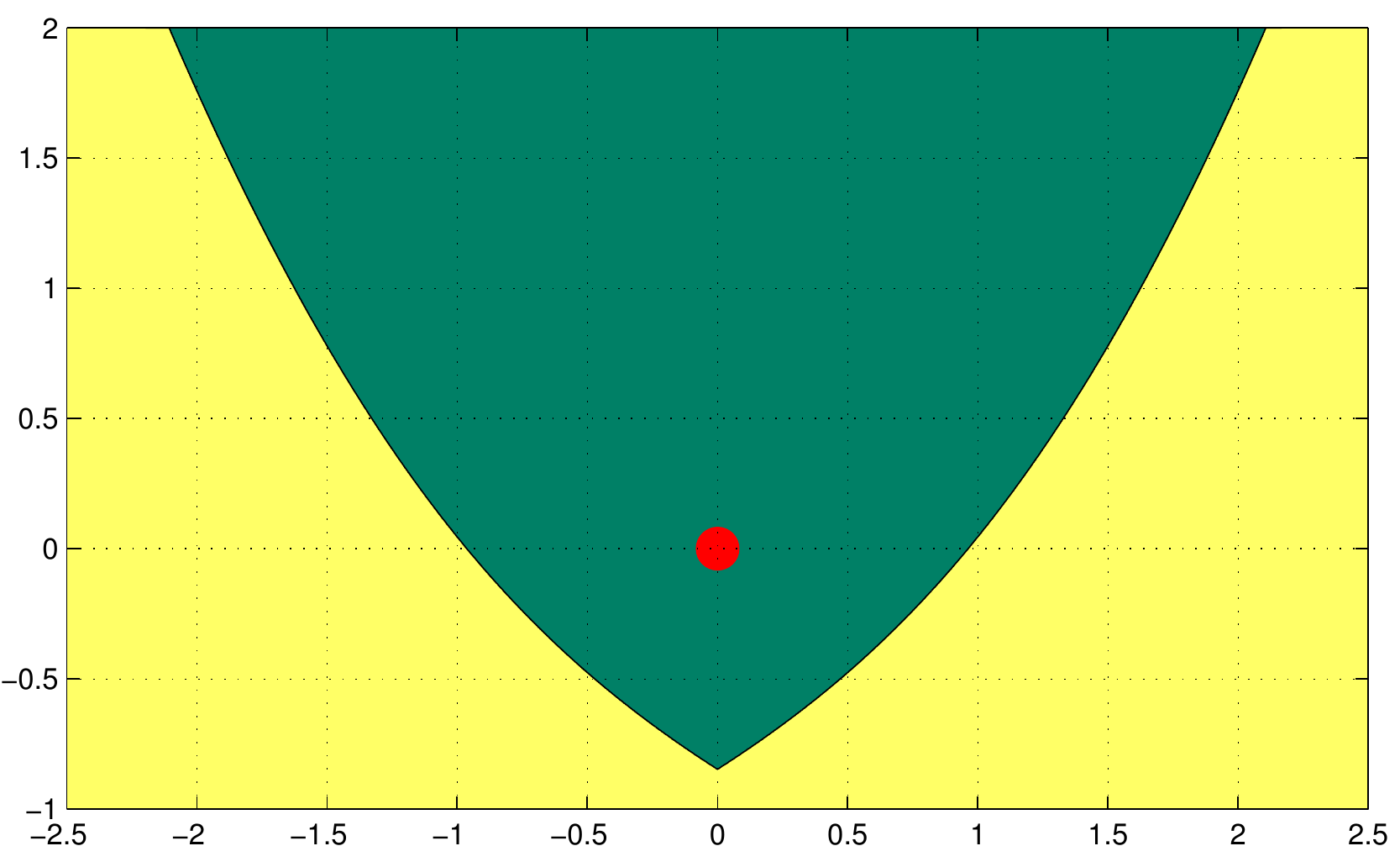}
            \put(76,-4){\tiny$\frac{q_1}{p\sqrt{\Temperature}}$}
            \put(-7,43){\tiny$\frac{\sigma_{11}}{p}$}
        \end{overpic}
    }
    \caption{\label{fig:hrBose} Hyperbolicity region of
    \eqref{eq:ms_5} with $\theta=-1$(Boson) for $\zz=0.1,0.5,0.9$. 
    \add{The blue part and the yellow part denote the hyperbolicity region
    and the non-hyperbolicity region, respectively, and the red point
    corresponds to the equilibrium.}}
\end{figure}

\begin{figure}
    \subfigure[$\zz=0.1$]{ 
        \begin{overpic}[width=.3\textwidth]{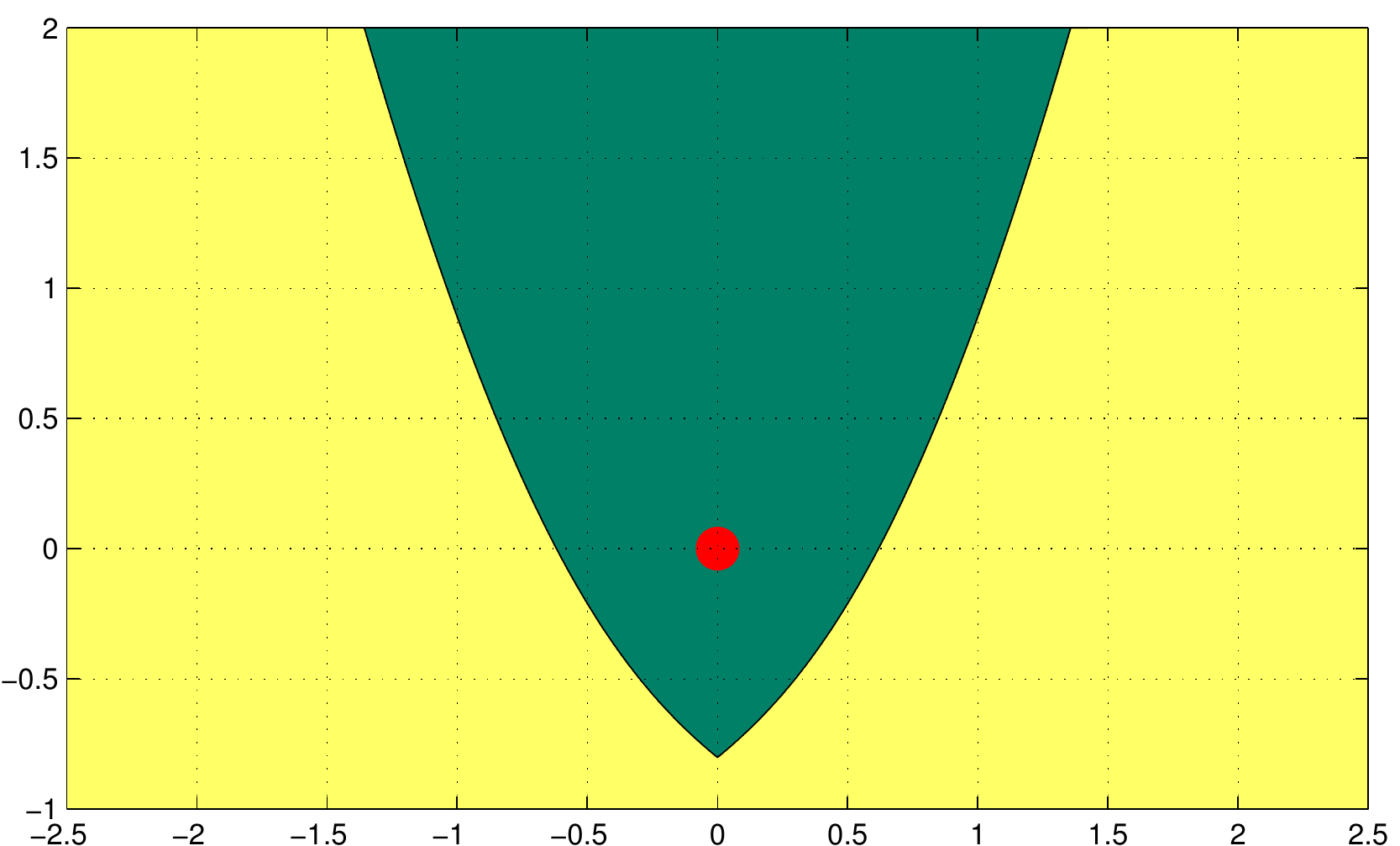}
            \put(76,-4){\tiny$\frac{q_1}{p\sqrt{\Temperature}}$}
            \put(-7,43){\tiny$\frac{\sigma_{11}}{p}$}
        \end{overpic}
    }
    \subfigure[$\zz=0.5$]{ 
        \begin{overpic}[width=.3\textwidth]{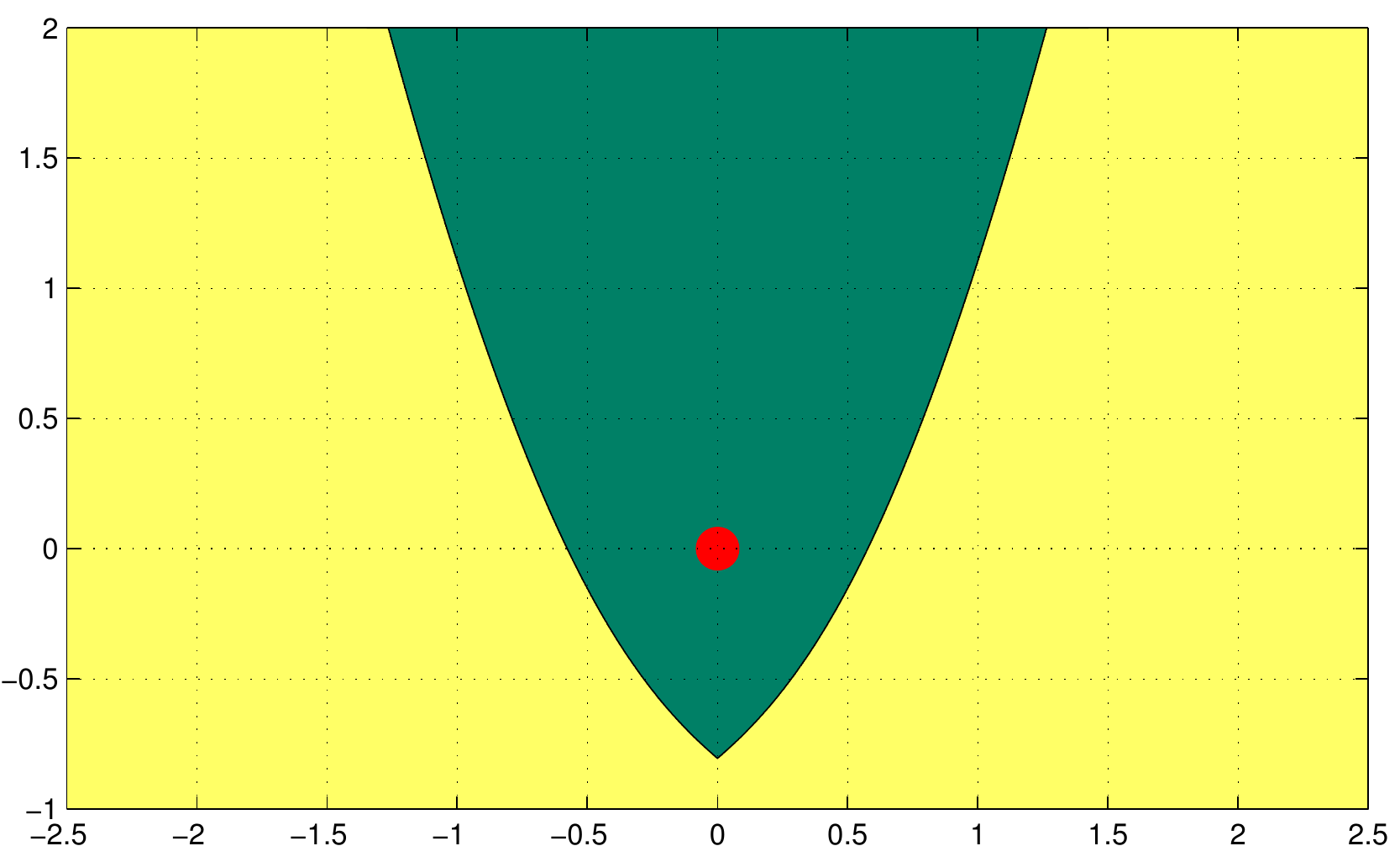}
            \put(76,-4){\tiny$\frac{q_1}{p\sqrt{\Temperature}}$}
            \put(-7,43){\tiny$\frac{\sigma_{11}}{p}$}
        \end{overpic}
    }
    \subfigure[$\zz=2$]{ 
        \begin{overpic}[width=.3\textwidth]{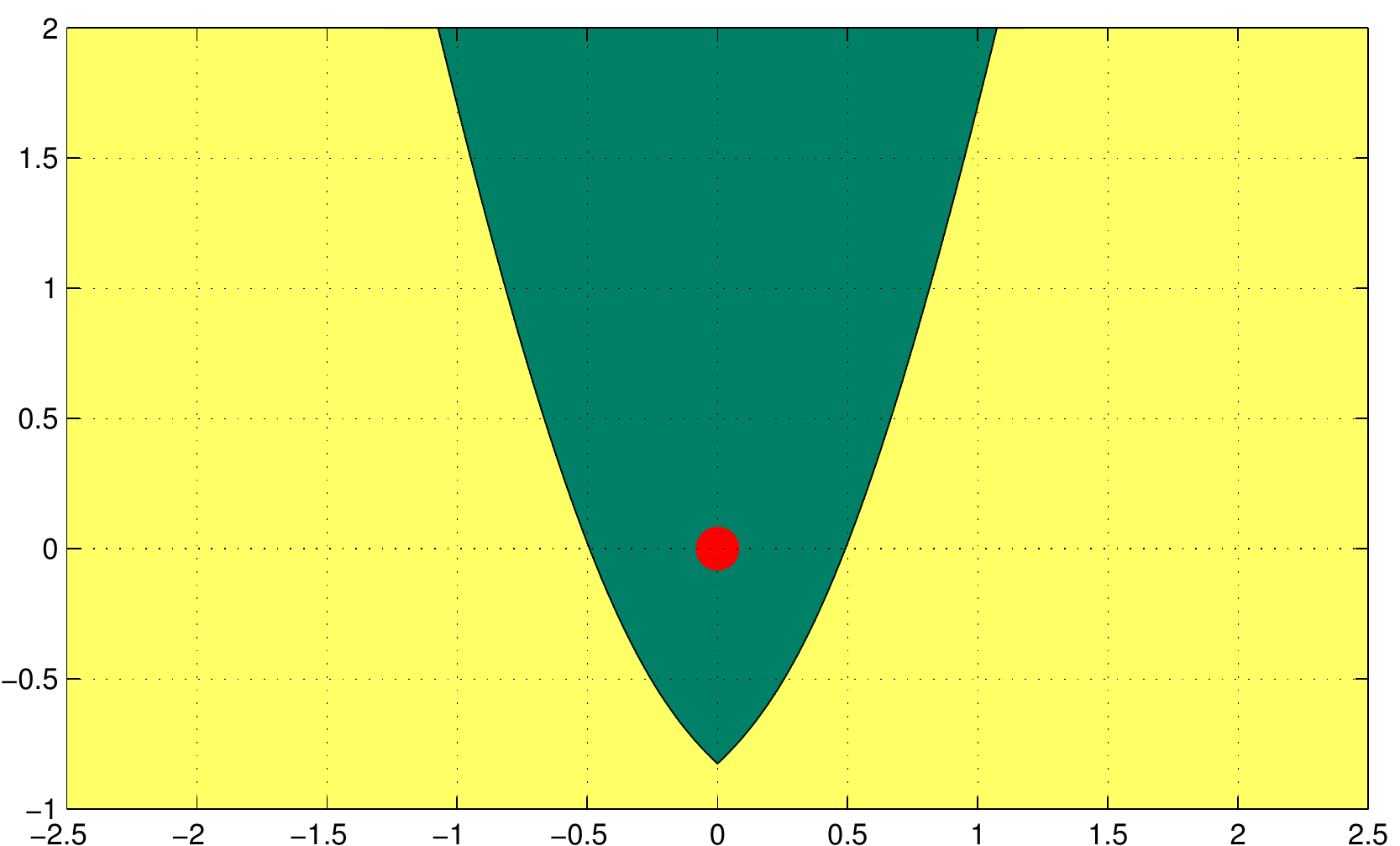}
            \put(76,-4){\tiny$\frac{q_1}{p\sqrt{\Temperature}}$}
            \put(-7,43){\tiny$\frac{\sigma_{11}}{p}$}
        \end{overpic}
    }
    \caption{\label{fig:hrFermi}
    Hyperbolicity region of \eqref{eq:ms_5} with $\theta=1$(Fermion)
    for $\zz=0.1,0.5,2$.
    \add{The blue part and the yellow part denote the hyperbolicity region
    and the non-hyperbolicity region, respectively, and the red point
corresponds to the equilibrium.}}
\end{figure}
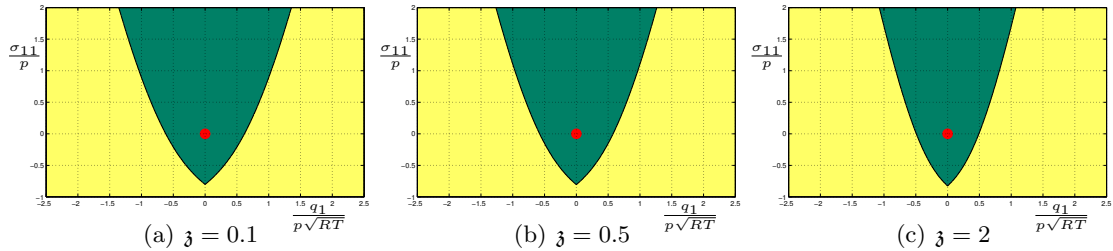

Similar to the Grad's 13-moment system for classical gasses
\cite{Muller}, for the quantum case, the thermodynamic equilibrium is
an interior point of the hyperbolicity region, but the system is not
hyperbolic for all admissible $\sigma_{11}$ and $q_1$. Additionally,
as shown in Figs. \ref{fig:hrBose} and \ref{fig:hrFermi}, the
hyperbolicity region expands as $\zz$ increasing for Boson, while the
hyperbolicity region shrinks as $\zz$ increasing for Fermion. 
Therefore, the hyperbolicity region depends on $\zz$. However, for the
classical Grad's 13-moment system, the hyperbolicity region does not
depend on the equilibrium. \add{It is the performance of the quantum
effect. It is well-known that for classical gas the sum of two
equilibrium distribution with the same mean velocity and temperature
is still an equilibrium, which does not hold any more for quantum
gas due to the quantum effect.
And as the fugacity decreases, the quantum effect weakens, which
corresponds $\Li_s/\zz \rightarrow 1$ as $\zz\rightarrow 0$ (for
classical gas, $\Li_s/\zz=1$) in mathematics. For Fermion, the
dimensionless coefficient of $\frac{q_1}{p\sqrt{\Temperature}}$ is
$\frac{\Lihalf{5}}{\Lihalf{3}}$, which increases as the fugacity $\zz$
increasing. So for given $\frac{q_1}{p\sqrt{\Temperature}}$ the
perturbation raised by the heat flux turns stronger as $\zz$
increasing, which results in the hyperbolicity region shrinks.
Analogously, for Boson, the hyperbolicity region expands as 
$\zz$ increasing due to $\frac{\Lihalf{5}}{\Lihalf{3}}$ decreasing.
Moreover, the hyperbolicity region is symmetric on the heat flux $q_1$
but not symmetric on the stress $\sigma_{11}$, since $q_1$ is an
odd order moment of $f$ with respect to $\xi_1$ while $\sigma_{11}$ is
an even order moment. Precisely, if we mirror the distribution
function $f$ on the $\xi_1$ direction, i.e
$f(\xi_1,\xi_2,\xi_3)\rightarrow f(-\xi_1,\xi_2,\xi_3)$, then $q_1$
switches its sign, i.e $q_1\rightarrow -q_1$ and $\sigma_{11}$ keeps
unchanged.
}

\subsection{Hyperbolicity for 3D case}\label{sec:hyperbolic3D}
As pointed out above, the quantum Grad's 13-moment system for 1D
flow is hyperbolic around the thermodynamic equilibrium, but not
hyperbolic for all admissible states. Intuitively, for 3D case, the
model would also be hyperbolic around the equilibrium. However, in
\cite{Grad13toR13}, the authors revealed the opposite fact, and pointed
out that the classical Grad's 13-moment system is not hyperbolic even
around the equilibrium. In the following, we investigate the
hyperbolicity of the quantum Grad's 13-moment system for 3D case.

\add{
\begin{figure}[ht]
    \subfigure[Boson]{ 
        \begin{overpic}[width=.3\textwidth]{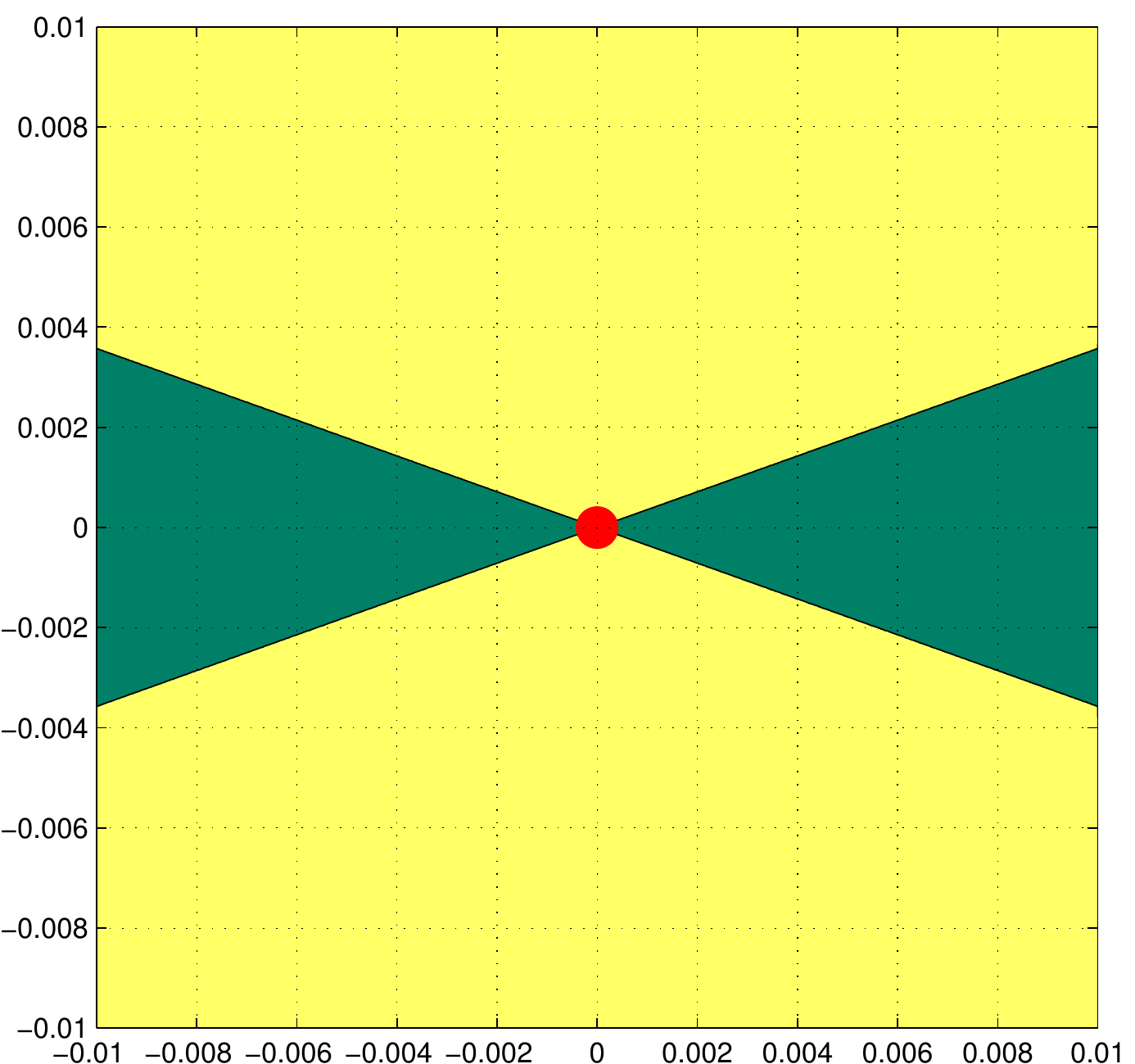}
            \put(76,-5){\tiny$\frac{q_1}{p\sqrt{\Temperature}}$}
            \put(-7,68){\tiny$\frac{\sigma_{12}}{p}$}
        \end{overpic}
    }
    \subfigure[Classical gas]{ 
        \begin{overpic}[width=.3\textwidth]{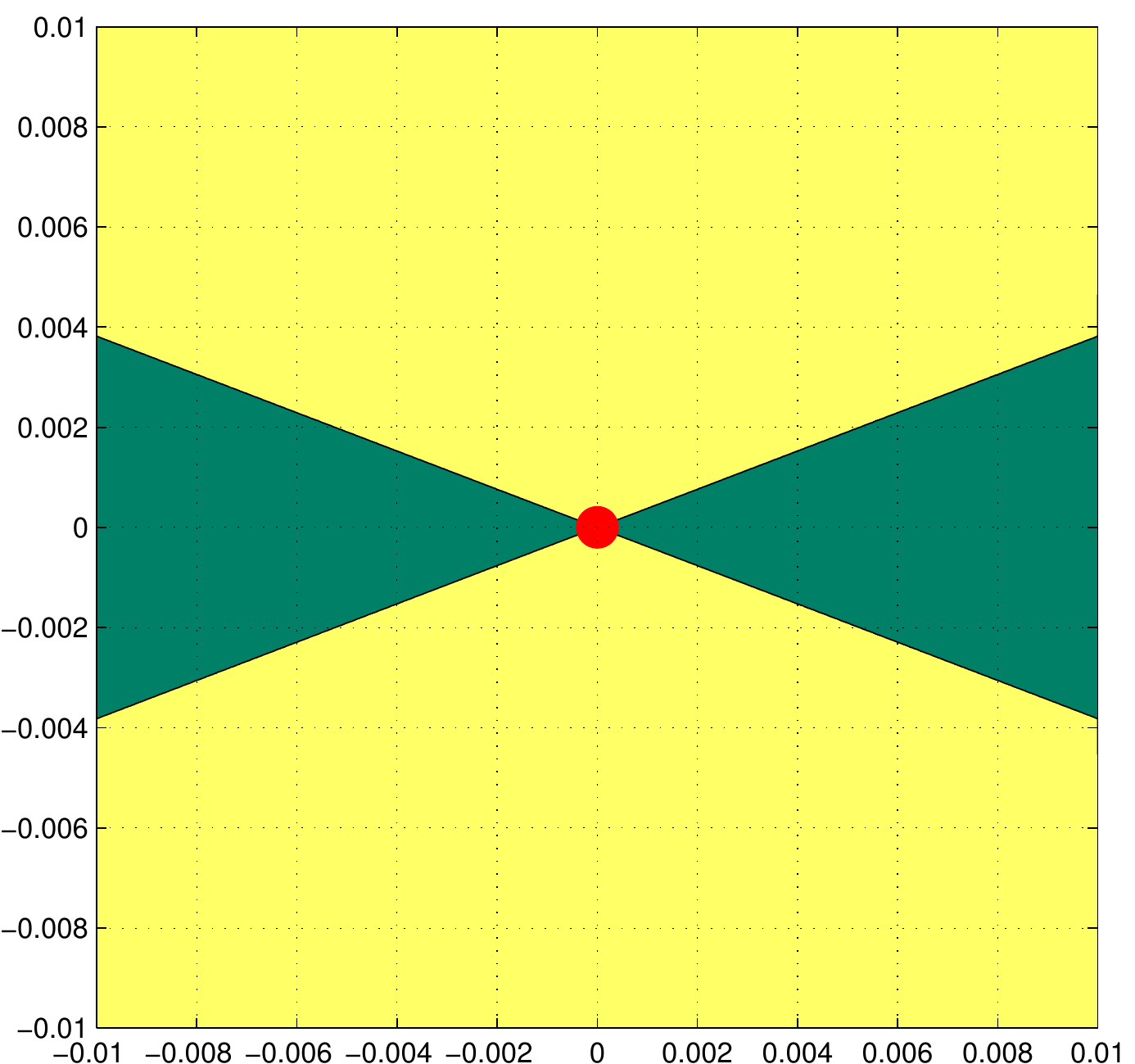}
            \put(76,-5){\tiny$\frac{q_1}{p\sqrt{\Temperature}}$}
            \put(-7,68){\tiny$\frac{\sigma_{12}}{p}$}
        \end{overpic}
    }
    \subfigure[Fermion]{ 
        \begin{overpic}[width=.3\textwidth]{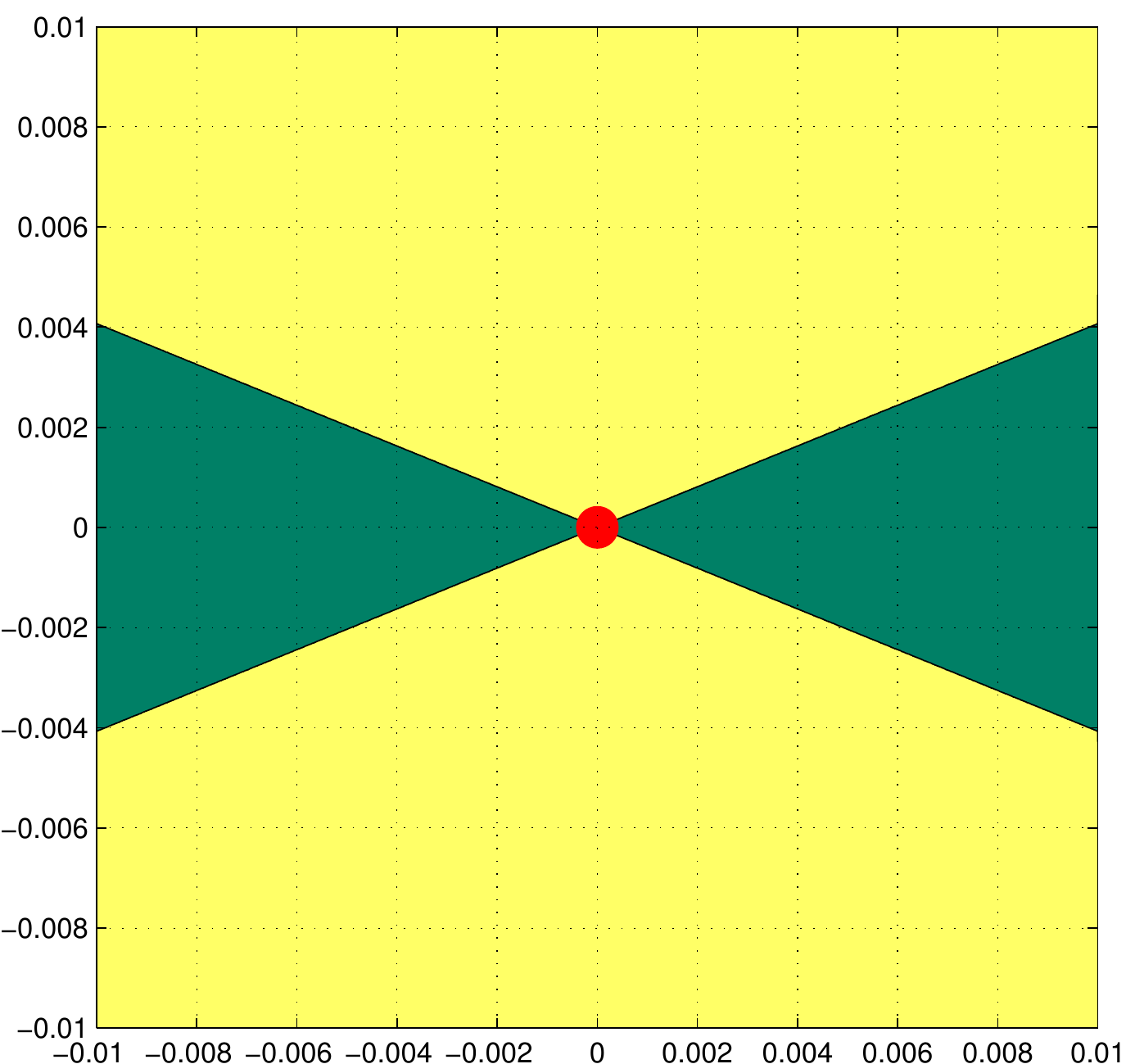}
            \put(76,-5){\tiny$\frac{q_1}{p\sqrt{\Temperature}}$}
            \put(-7,68){\tiny$\frac{\sigma_{12}}{p}$}
        \end{overpic}
    }
    \caption{\label{fig:hr13}
    A cross section of diagonalizable region of $\bA_1$ in
    \eqref{eq:ms13} with $\zz=0.5$ for $\theta=-1$ (Boson), $0$
    (Classical gas) and $1$ (Fermion).
    \add{The blue part and the yellow part denote the diagonalizable region
    and the non-diagonalizable region, respectively, and the red point
corresponds to the equilibrium.}}
\end{figure}
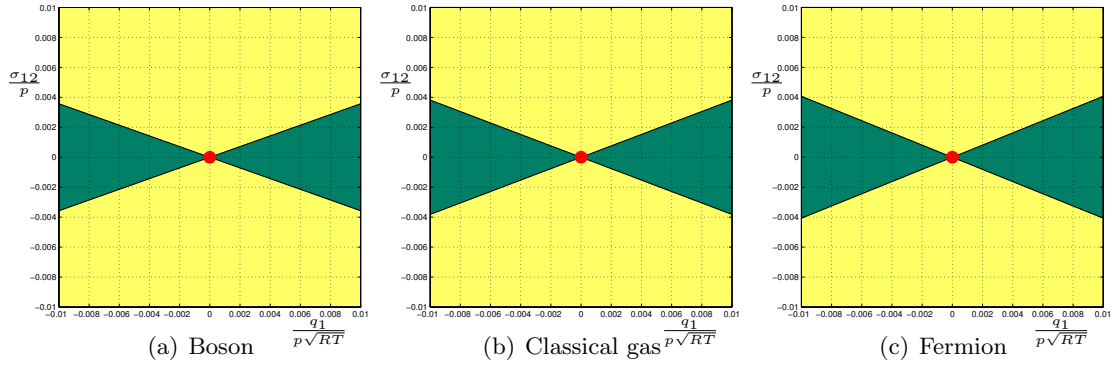
}

Let $\bw=(\rho,u_1,u_2,u_3,p_{11},p_{12},p_{13},p_{22},p_{23},p_{33},
q_1,q_2,q_3)^T$, then the quantum Grad's 13-moment system can be written
as 
\begin{equation}\label{eq:ms13}
    \pd{\bw}{t}+\bA_d\pd{\bw}{x_d}=\boldsymbol{Q},
\end{equation}
where $\boldsymbol{Q}=(0,0,0,0,Q_{11}^{(2)},Q_{12}^{(2)},Q_{13}^{(2)},
Q_{22}^{(2)},Q_{23}^{(2)},Q_{33}^{(2)},Q_1^{(3)},Q_2^{(3)},Q_3^{(3)})^T$,
and $\bA_d$ can be obtained directly from \eqref{eq:EulerEqs} and
\eqref{eq:sigmaq}.

If $\bw$ represents the thermodynamic equilibrium, i.e. 
\[
    p_{ij}=p\delta_{ij}, ~ q_i=0,\quad i,j=1,2,3,
\]
\revised{using the same technique as that in
\mbox{\cite{Grad13toR13}}, we can prove $\bA_d$ are diagonalizable.} 
{then the characteristic polynomial of $\bA_1$ can be directly
calculated as
\begin{equation}\label{eq:cp_13eq}
    |\lambda
    \identity-\bA_1|=(\Temperature)^{13/2}\hat{\lambda}^5\left(
    \hat{\lambda}^2-\frac{7\Lihalf{9}}{5\Lihalf{7}} \right)^2
    \left( \hat{\lambda}^4-c_1\hat{\lambda}^2+ c_0\right),
\end{equation}
where $\hat{\lambda}=(\lambda-u_1)/\sqrt{\Temperature}$, and $c_0$ and
$c_1$ and quantities depending on $\zz$, and satisfying $c_0>0$ and
$c_1>4c_0$ (see Appendix \ref{app:coe} for their concrete forms). 
Hence, in this case, the matrix $\bA_1$ has $13$ real eigenvalues
$\lambda_i=u_1+\sqrt{\Temperature}\hat{\lambda}_i$, $i=1,\cdots,13$
with
\begin{equation}\label{eq:eigenvaluesG13}
    \begin{aligned}
        \hat{\lambda}_{1,13}&=\pm\sqrt{\frac{c_1+\sqrt{c_1^2-4c_0}}{2}},\qquad
        \hat{\lambda}_{2,3,11,12}=\pm\sqrt{\frac{7\Lihalf{9}}{5\Lihalf{7}}},\\
        \hat{\lambda}_{4,10}&=\pm\sqrt{\frac{c_1-\sqrt{c_1^2-4c_0}}{2}},\qquad
        \hat{\lambda}_{5,6,7,8,9}=0.
    \end{aligned}
\end{equation}
Furthermore, we point out that $\bA_1$ and any linear combination of
$\bA_d$ are diagonalizable. The proof will be given in Sect.
\ref{sec:hyperbolicityproof} and Sect. \ref{sec:discussion}.}  Next,
we will show that if $\bw$ is perturbed around the equilibrium,
$\bA_1$ may be not diagonalizable anymore.  Consider the case
\[
    p_{11}=p_{22}=p_{33}=p, \quad 
    p_{13}=p_{23}=0,\quad 
    q_1=q_2=q_3=0,\quad 
    p_{12}=\epsilon p.
\]
The thermodynamic equilibrium corresponds to $\epsilon=0$. The
characteristic polynomial of $\bA_1$ is 
\begin{equation}
    |\lambda\identity-\bA_1|=\frac{\Temperature^{13/2}}{125}
    \hat{\lambda}^3\left( 5\hat{\lambda}^2-7\frac{\Lihalf{7}}{\Lihalf{5}} \right)
    g(\hat{\lambda}^2),
\end{equation}
where $\hat{\lambda}=(\lambda-u_1)/\sqrt{\Temperature}$, and 
\begin{equation}\label{eq:def_g}
    \begin{aligned}
        g(x)=25x^4&+c_4x^3+c_3x^2+c_2x\\
        &-14\epsilon^2\frac{\Lihalf{5}^2
        \left( 
        14\Lihalf{3}\Lihalf{7}\Lihalf{9}^2+35\Lihalf{5}^2\Lihalf{9}^2
        -80\Lihalf{5}\Lihalf{7}^2\Lihalf{9}+35\Lihalf{7}^4\right)}
        {\Lihalf{3}\Lihalf{7}^3\left( 5\Lihalf{1}\Lihalf{5}-3\Lihalf{3}^2
    \right)},
    \end{aligned}
\end{equation}
where $c_2,c_3,c_4$ are some quantities depending on $\zz$ \add{and
$\epsilon$ (see Appendix \ref{app:coe} for its concrete forms)}. If
$\epsilon\neq0$, then the constant term of $g(x)$ 
\[
    -14\epsilon^2\frac{\Lihalf{5}^2
    \left( 
    14\Lihalf{3}\Lihalf{7}\Lihalf{9}^2+35\Lihalf{5}^2\Lihalf{9}^2
    -80\Lihalf{5}\Lihalf{7}^2\Lihalf{9}+35\Lihalf{7}^4\right)}
    {\Lihalf{3}\Lihalf{7}^3\left( 5\Lihalf{1}\Lihalf{5}-3\Lihalf{3}^2 \right)}<0.
\]
So $g(x)$ has at least one negative root, which indicates $\bA_1$ has
at least two complex eigenvalues, i.e. $\bA_1$ is not real
diagonalizable. Consequently, for arbitrarily small $\epsilon\neq0$, the system
\eqref{eq:ms13} is not hyperbolic. The thermodynamic equilibrium is
not an interior point of the hyperbolicity region of the quantum
Grad's 13-moment system. Whether $\bA_1$ is diagonalizable depends on
several variables, and we cut a cross section of the region crossing
the thermodynamic equilibrium. For a given $\zz$, set
$\sigma_{11}=\sigma_{22}=\sigma_{33}=\sigma_{13}=\sigma_{23}=0$,
$q_2=q_3=0$, and plot the cross sections in Fig. \ref{fig:hr13}.
\add{It is reminded that since $\sigma_{12}$ is
an odd order moment of $f$ with respect to $\xi_1$, the diagonalizable
regions are symmetric on both $q_1$ and $\sigma_{12}$.}
One can observe that the equilibrium is just on the boundary of the
real diagonalizable region of $\bA_1$. Therefore, the quantum Grad's
13-moment system is not hyperbolic in any neighborhood of the
equilibrium. Without the hyperbolicity in a neighborhood of the
equilibrium, the well-posedness of the quantum Grad's 13-moment system
is lost even if the phase density is extremely close to the
equilibrium. Although in \cite{QuantumGrad13}, the author claimed that
the quantum Grad's 13-moment system possesses a lot of good
properties, this fatal drawback tremendously limits the practical
application of the system.

\section{Regularization of the 13-Moment system}
\label{sec:regularization}
In this section, we try to regularize the quantum Grad's 13-moment
system for the QBE to a globally hyperbolic model.
\add{Hyperbolicity is a critical condition for the existence and
stability of the solution of a model. However, it is not enough to
guarantee the model is a good approximation of the QBE. Here
we concern two criteria: 
\begin{enumerate}
    \item the linearized equations of the regularized system are same
        as those of the quantum Grad's 13-moment system; 
    \item contains the correct NSF law. 
\end{enumerate}
Since the equilibrium plays an
essential role in Grad's expansion, it is expected that at the
equilibrium, the regularized system shares the same wave structures,
i.e. the coefficients matrix has the eigenvalues
\eqref{eq:eigenvaluesG13} and possesses same eigenvectors as those of
Grad's 13-moment system, which is equivalent to the first criterion.
The second criterion is natural because Grad's moment system contains
the NSF law.

In \cite{framework, Fan2015}, the authors proposed a framework on
moment model reduction from kinetic equations to globally hyperbolic
hydrodynamic system, by restudied their previous work \cite{Fan,
Fan_new}. A natural way is to directly apply the framework on
the quantum Grad's 13-moment system to obtain a globally hyperbolic
system. However, the resulting system fails to satisfy the two
criteria. Given the important role of the equilibrium in Grad's
expansion, we split the expansion into the equilibrium part and the
non-equilibrium part, and only apply the framework on the
non-equilibrium. The resulting system is globally hyperbolic and also
consistent with the two criteria. Further study shows the essential of the
regularization is decoupling the derivative operator and the
multiplying velocity operator.}
\delete{using the framework on the model reduction for generic kinetic
equations in \mbox{\cite{framework, Fan2015}}.  It should be noted that the
framework in \mbox{\cite{framework, Fan2015}} can not be applied to the
quantum Boltzmann equation directly.  The equilibrium part and the
non-equilibrium part will be dealt with in different ways below.}

\subsection{Revisit of Grad's 13-moment system}
\add{
Let us examine the quantum Grad's 13-moment system in further details
at first. Given $\zz$, $\bu$ and $\Temperature$, we have a weight
function denoted as
\begin{equation}
    \weightzut(\bv) =\frac{1}{\zz^{-1}\exp\left(
    \frac{|\bv-\bu|^2}{2\Temperature} \right)+\theta}.
\end{equation}
Then we denote $\bbH$ be the weighted $L^2$ function space with the
inner product
\begin{equation}\label{eq:innerproduct}
    \langle g, h\rangle^{[\zz,\bu,\Temperature]} = \frac{m}{\hat{h}^3}\int_{\bbR^3}\frac{gh}{\weightzut}\dd\bv.
\end{equation}
Let us consider the distribution functions satisfying \eqref{eq:QBE}
in the Hilbert space $\bbH$.  Define a set of basis functions
\[
    \begin{aligned}
        \phi^{(0)}&=\frac{1}{\rho}\weightzut(\bv),
        & \phi_{ij}^{(2)}&= \weightzut
        \frac{1}{2p}\frac{\Lihalf{5}}{\Lihalf{7}}
        \left( \frac{c_ic_j}{\Temperature}-\delta_{ij}\frac{\Lihalf{5}}{\Lihalf{3}} \right),\\
        \phi_i^{(1)}&=\weightzut(\bv)\frac{c_i}{p},
        &
        \phi_i^{(3)}&= \weightzut
        \frac{\mathfrak{B}^{-1}}{5p\Temperature}c_i\left(
        \frac{|\bc|^2}{\Temperature}-5\frac{\Lihalf{7}}{\Lihalf{5}}\right),
    \end{aligned}
\]
where $c_i=v_i-u_i$,
$\rho=\constantsymbol\sqrt{2\pi\Temperature}^3\Lihalf{3}$,
$p=\constantsymbol\sqrt{2\pi\Temperature}^3\Temperature\Lihalf{5}$,
and
$\mathfrak{B}=7/2(\Lihalf{9}/\Lihalf{5})-5/2(\Lihalf{7}^2/\Lihalf{5}^2)$.
Let the sub-space of $\bbH$, spanned by $\phi^{(0)}$, $\phi^{(1)}_i$,
$\phi^{(2)}_{ij}$, and $\phi^{(3)}_i$, to be denoted as
\begin{equation}
    \bbHzut=\rmspan\left\{\phi^{(0)}, \phi^{(1)}_i, \phi^{(2)}_{ij},
    \phi^{(3)}_i\right\}.
\end{equation}
Since the basis functions are product of the weight function and
polynomials, we have
\begin{equation}
    \bbHzut=\rmspan\left\{\weightzut,\weightzut v_i, \weightzut v_iv_j,
    \weightzut |\bv|^2v_i \right\}.
\end{equation}
Let $\mP$ to be the natural projection from $\bbH$ to $\bbHzut$
defined as
\begin{equation}
    \begin{array}{cccl}
        & \bbH & \rightarrow & \bbHzut\\
        \mP: & g & \rightarrow & \mG =
        a^{(0)}\phi^{(0)}+a^{(1)}_i\phi^{(1)}_i
        +a^{(2)}_{ij}\phi^{(2)}_{ij}
        +a^{(3)}_i\phi^{(3)}_i
    \end{array},
\end{equation}
where the coefficients $a^{(0)}$, $a^{(1)}_i$, $a^{(2)}_{ij}$ and
$a^{(3)}_i$ are defined as
\begin{equation}
    \begin{aligned}
        a^{(0)}&=\rho\langle g\phi^{(0)}\rangle^{[\zz,\bu,\Temperature]},&
        a^{(1)}_i&=p\langle g\phi^{(1)}_i\rangle^{[\zz,\bu,\Temperature]},\\
        a^{(2)}_{ij}&=\frac{2p\Lihalf{7}}{\Lihalf{5}}\langle
        g\phi^{(2)}_{ij}\rangle^{[\zz,\bu,\Temperature]} - b\delta_{ij},&
        a^{(3)}_i&=5p(\Temperature)^2\mathfrak{B}\langle g\phi^{(3)}_i\rangle^{[\zz,\bu,\Temperature]},\\
    \end{aligned}
\end{equation}
where
$b=(1-\mathfrak{b}) \constantsymbol \displaystyle \int_{\bbR^3}g\left(
  \frac{1}{3}|\bc|^2-\frac{\Lihalf{5}}{\Lihalf{3}}\Temperature\right)\dd\bv$,
and $\mathfrak{b}=\frac{2}{5-3\Lihalf{5}^2/(\Lihalf{3}\Lihalf{7})}$.
Obviously, the projection $\mP$ is an orthogonal projection on $\bbH$
and is identical on $\bbHzut$, i.e., $\mP \mG=\mG$ for
$\forall \mG\in\bbHzut$.}  We recall that $\zz$, $\bu$ and
$\Temperature$ are given parameters, and the projection $\mP$ is
dependent on these parameters.  Particularly, if $\zz$, $\bu$ and
$\Temperature$ are defined as
\begin{equation*}
    \int_{\bbR^3}g\dd\bv=\sqrt{2\pi\Temperature}^3\Lihalf{3},\quad
    \int_{\bbR^3}gv_i\dd\bv=\sqrt{2\pi\Temperature}^3\Lihalf{3}u_i,\quad
    \int_{\bbR^3}g|\bc|^2\dd\bv=3\sqrt{2\pi\Temperature}^3\Temperature\Lihalf{5},
\end{equation*}
and denoted by $\zz=\zz(g)$, $\bu=\bu(g)$ and
$\Temperature=\Temperature(g)$, then we have 
\begin{equation}
    a^{(0)} = \rho,\quad 
    a^{(1)}_i = 0,\quad
    a^{(2)}_{ij} =\sigma_{ij},\quad
    a^{(3)}_i =q_i.
\end{equation}
Then for a distribution $f\in\bbH$, if $\zz=\zz(f)$, $\bu=\bu(f)$ and
$\Temperature=\Temperature(f)$, $\mP f$ is Grad's expansion
\eqref{eq:Grad13expansion1}, \add{i.e $f_{G13} = \mP f$}. Therefore,
we need only to study the distribution in the space $\bbHzut$ in the
following.

\add{
Actually, Grad's moment system is derived by
\begin{equation}\label{eq:deriveGrad13}
    \frac{m}{\hat{h}^3}\int_{\bbR^3}\psi\pd{f_{G13}}{t}\dd\bv 
    +\frac{m}{\hat{h}^3}\int_{\bbR^3}\psi v_d\pd{f_{G13}}{x_d}\dd\bv
    =\frac{m}{\hat{h}^3}\int_{\bbR^3}\psi Q(f_{G13},f_{G13})\dd\bv,
\end{equation}
where $\psi=1,v_i,v_iv_j,v_i|\bv|^2$.
Noticing $f_{G13} = \mP f$, Eq. \eqref{eq:deriveGrad13} is equivalent
to
\begin{equation}
    \left\langle\phi, \pd{\mP f}{t}\right\rangle^{[\zz,\bu,\Temperature]}
    +\left\langle\phi, v_d\pd{\mP f}{x_d}\right\rangle^{[\zz,\bu,\Temperature]}
    =\left\langle\phi, Q(\mP f, \mP f)\right\rangle^{[\zz,\bu,\Temperature]},
    \qquad \phi\in\bbHzut.
\end{equation}
Since $\mP$ is orthogonal, Grad's 13-moment system is
essentially equivalent to
\begin{equation}\label{eq:BoltzmannmP}
    \mP\pd{\mP f}{t}+\mP v_d\pd{\mP f}{x_d} = \mP Q(\mP f, \mP f).
\end{equation}
For the time derivative part, direct calculation yields
\begin{equation}\label{eq:timedeverivate}
    \begin{aligned}
        \mP\pd{\mP f}{t} &=
        \pd{\rho}{t}\phi^{(0)} +\rho\pd{u_i}{t}\phi^{(1)}_i
        +\left( 
        \pd{\sigma_{ij}}{t}+\delta_{ij}\mathfrak{b}\left(
        \pd{p}{t}-\frac{p}{\rho}\pd{\rho}{t} \right)
        \right)\phi_{ij}^{(2)} \\
        &\quad+
        \left( \pd{q_i}{t}+\sigma_{ij}\pd{u_j}{t}+\frac{5}{2}\left(
        p-\rho\Temperature\frac{\Lihalf{7}}{\Lihalf{5}} \right)
        \pd{u_i}{t} \right)\phi_i^{(3)}.
    \end{aligned}
\end{equation}
}
We remark that the upper equation is also valid for the derivatives of
$x_d$. Next, we calculate the convection term:
\begin{equation}
    \begin{aligned}
       \mP v_d\pd{\mP f}{x_d} &=\mP(u_d+c_d)\pd{\mP f}{x_d}
       =u_d\mP\pd{\mP f}{x_d} + \mP c_d\pd{\mP f}{x_d}.
    \end{aligned}
\end{equation}
Direct calculation and simplification yield
\begin{equation}\label{eq:convectionterm}
    \begin{aligned}
        \mP v_d\pd{\mP f}{x_d} &= u_d\mP\pd{\mP f}{x_d}
        +\rho\pd{u_d}{x_d}\phi^{(0)}+\pd{p_{id}}{x_d}\phi_i^{(1)}\\
        &\quad{\tiny
            +\left( 
            \sigma_{ij}\pd{u_d}{x_d}+2p_{d\langle j}\pd{u_{i\rangle}}{x_d}
            +\frac{4}{5}\pd{q_{\langle i}}{x_{j\rangle}}
            \pd{\sigma_{ij}}{t}
            +\frac{2\delta_{ij}\mathfrak{b}}{3}
            \left( p_{id}\pd{u_i}{x_d}+\pd{q_d}{x_d} \right) \right)
            \phi_{ij}^{(2)}
        }\\
        &\quad+
        \left( 
        \frac{7}{5}q_i\pd{u_d}{x_d}+\frac{7}{5}q_d\pd{u_i}{x_d}+\frac{2}{5}q_k\pd{u_k}{x_i}
        +\frac{1}{2}\pd{\Delta_{id}}{x_d}
        -\frac{5}{2}\Temperature\frac{\Lihalf{7}}{\Lihalf{5}}\pd{p_{id}}{x_d}
        \right)
        \phi_i^{(3)}.
    \end{aligned}
\end{equation}
Because of \eqref{eq:collisioninvariants} and the definitions of
$Q_{ij}^{(2)}$ and $Q_i^{(3)}$, the collision part can be written as
\begin{equation}\label{eq:collisionterm}
    \mP Q(f,f)=
    Q_{ij}^{(2)}\phi_{ij}^{(2)}+Q_i^{(3)}\phi_i^{(3)}
\end{equation}
and $Q_{ii}^{(2)}=0$.

Collecting \eqref{eq:timedeverivate}, \eqref{eq:convectionterm},
\eqref{eq:collisionterm}, and matching the coefficients of
$\phi^{(0)}$, $\phi_i^{(1)}$, $\phi_{ij}^{(2)}$ and $\phi_i^{(3)}$, we
obtain the moment system
\begin{subequations}\label{eq:BFGrad13}
    \begin{align}
        \opd{\rho}{t}&+\rho\pd{u_d}{x_d}=0,\\
        \rho\opd{u_i}{t}&+\pd{p_{id}}{x_d}=0,~i=1,2,3,\\
        \begin{split}
            \opd{\sigma_{ij}}{t}&+\sigma_{ij}\pd{u_d}{x_d}+2p_{d\langle
            j}\pd{u_{i\rangle}}{x_d}
            +\frac{4}{5}\pd{q_{\langle i}}{x_{j\rangle}}\\
            &+\delta_{ij}\mathfrak{b}
            \left(  \opd{p}{t}-\frac{p}{\rho}\opd{\rho}{t} 
            + \frac{2}{3} \left( p_{id}\pd{u_i}{x_d}+\pd{q_d}{x_d} \right)
            \right)
            =Q_{ij}^{(2)},
        \end{split}\\
        \begin{split}\label{eq:BFGrad13-q}
            \opd{q_i}{t}&+\frac{5}{2}\left(p-\rho\Temperature\frac{\Lihalf{7}}{\Lihalf{5}}\right)\opd{u_i}{t}+\sigma_{ij}\opd{u_j}{t}
            +\frac{7}{5}q_i\pd{u_d}{x_d}+\frac{7}{5}q_d\pd{u_i}{x_d}+\frac{2}{5}q_k\pd{u_k}{x_i}\\
            &+\frac{1}{2}\pd{\Delta_{id}}{x_d}
            -\frac{5}{2} \Temperature\frac{\Lihalf{7}}{\Lihalf{5}}
            \pd{p_{id}}{x_d}
            =Q_i^{(3)},
        \end{split}
    \end{align}
\end{subequations}
which is exactly the same as the quantum Grad's 13-moment system in
Sect. \ref{sec:Grad13}. 

\add{
The projection $\mP$ is an alternative representation of Grad's expansion
\eqref{eq:Grad13expansion1}, and Eq. \eqref{eq:BoltzmannmP} is an
equivalent form of \eqref{eq:deriveGrad13}. The new formula helps us
understand the relationship between the QBE and Grad's moment system,
and is also useful in the following part of this section.
}

\subsection{A trivial regularized 13-moment system}\label{sec:hme13framewrok}
\add{ In Ref. \cite{Fan2015}, the authors pointed out that the
  underlying reason due to the loss of hyperbolicity of Grad's moment
  system is that Grad's moment method treats the time derivative and
  space derivatives in different ways. Precisely, we may see this
  point by comparison of the QBE \eqref{eq:QBE} and
  \eqref{eq:BoltzmannmP} in details. In \eqref{eq:BoltzmannmP},
  $\mP \pd{}{t}$ is the corresponding time derivative operator
  $\pd{}{t}$ in \eqref{eq:QBE}, and $\mP v_d \pd{}{x_d}$ is derived
  from the convection operator $v_d\pd{}{x_d}$ in
  \eqref{eq:QBE}. Notice that $v_d\pd{}{x_d}$ can be split into the
  multiplication by velocity and the spatial derivative operator,
  while the term $\mP v_d \pd{}{x_d}$ has no similar splitting
  anymore. Hence, in \eqref{eq:BoltzmannmP}, the time and space
  derivatives are treated in different ways. The authors of
  \cite{Fan2015} suggest split the convection operator
  $\mP v_d\pd{}{x_d}$ as $\mP v_d\mP\pd{}{x_d}$, which indicates the
  equivalent equation of the moment system turns to
\begin{equation}\label{eq:BoltzmannmPP}
    \mP\pd{\mP f}{t}+\mP v_d\mP\pd{\mP f}{x_d} = \mP Q(\mP f, \mP f).
\end{equation}
The corresponding moment system is guaranteed to be globally
hyperbolic by Thm. 1 in \cite{Fan2015} or Thm. 4.1 in
\cite{framework}.
}

Since \eqref{eq:timedeverivate} is also valid for space derivatives,
direct calculation yields
\begin{small}
\begin{equation}\label{eq:convectionterm_regularize}
    \begin{aligned}
        \mP v_d\mP\pd{\mP f}{x_d} &= u_d\mP\pd{\mP f}{x_d}
        +\rho\pd{u_d}{x_d}\phi^{(0)}+\pd{p_{id}}{x_d}\phi_i^{(1)}\\
        &\quad{
            +\left( 
            2p\pd{u_{\langle i}}{x_{l\rangle}}
            +\frac{4}{5}\pd{q_{\langle i}}{x_{j\rangle}}
            +\frac{4}{5}\sigma_{k\langle i}\pd{u_{k}}{x_{j\rangle}}
            + \frac{2\delta_{ij}\mathfrak{b}}{3} \left( p_{id}\pd{u_i}{x_d}+\pd{q_d}{x_d} \right) \right)
            \phi_{ij}^{(2)}
        }\\
        &\quad+
        \left( 
        \left( \frac{7\Lihalf{9}}{2\Lihalf{7}}-\frac{5\Lihalf{7}}{2\Lihalf{5}} \right)
        \pd{\sigma_{id}}{x_d}
        +\frac{5}{2}\left(
        \frac{7\frac{\Lihalf{7}}{\Lihalf{5}}-3\frac{\Lihalf{3}}{\Lihalf{1}}}
        {5-3\frac{\Lihalf{3}^2}{\Lihalf{5}\Lihalf{1}}}-\frac{\Lihalf{7}}{\Lihalf{5}} \right)
        \left( \pd{p}{x_i}-\frac{p}{\rho}\pd{\rho}{x_i} \right)
        \right)\Temperature
        \phi_i^{(3)}.
    \end{aligned}
\end{equation}
\end{small}
Collecting \eqref{eq:timedeverivate}, \eqref{eq:convectionterm_regularize} and
\eqref{eq:collisionterm}, and matching the coefficients of
$\phi^{(0)}$, $\phi_i^{(1)}$, $\phi_{ij}^{(2)}$ and $\phi_i^{(3)}$, we
obtain the regularized 13-moment system:
\begin{subequations}\label{eq:regularized13_framework}
    \begin{align}
        \opd{\rho}{t}&+\rho\pd{u_d}{x_d}=0,\\
        \rho\opd{u_i}{t}&+\pd{p_{id}}{x_d}=0,~i=1,2,3,\\
        \begin{split}
            \opd{\sigma_{ij}}{t}&
            +2p\pd{u_{\langle i}}{x_{j\rangle}}
            +\frac{4}{5}\pd{q_{\langle i}}{x_{j\rangle}}
            +\frac{4}{5}\sigma_{d\langle i}\pd{u_{d}}{x_{j\rangle}}\\
            &+\delta_{ij}\mathfrak{b}
            \left(  \opd{p}{t}-\frac{p}{\rho}\opd{\rho}{t} 
            + \frac{2}{3} \left( p_{kd}\pd{u_k}{x_d}+\pd{q_d}{x_d} \right)
            \right)
            =Q_{ij}^{(2)},
        \end{split}\\
        \begin{split}\label{eq:regularized13_framework-q}
            \opd{q_i}{t}&
            +\frac{5}{2}\left(
            p-\rho\Temperature\frac{\Lihalf{7}}{\Lihalf{5}} \right)\opd{u_i}{t}
            +\sigma_{ij}\opd{u_j}{t} 
            +\Temperature \left( \frac{7\Lihalf{9}}{2\Lihalf{7}}-\frac{5\Lihalf{7}}{2\Lihalf{5}} \right) \pd{\sigma_{id}}{x_d}
            \\
            &+\frac{5}{2}\Temperature\left(
            \frac{7\frac{\Lihalf{7}}{\Lihalf{5}}-3\frac{\Lihalf{3}}{\Lihalf{1}}}
            {5-3\frac{\Lihalf{3}^2}{\Lihalf{5}\Lihalf{1}}}-\frac{\Lihalf{7}}{\Lihalf{5}} \right)
            \left( \pd{p}{x_i}-\frac{p}{\rho}\pd{\rho}{x_i} \right)
            =Q_i^{(3)}.
        \end{split}
    \end{align}
\end{subequations}

\add{
Comparing the quantum Grad's 13-moment system \eqref{eq:BFGrad13} and
the regularized 13-moment system \eqref{eq:regularized13_framework},
one can observe that only the governing equations of $\sigma_{ij}$ and
$q_i$ are different, and the two moment systems share the same
governing equations of $\rho$, $u_i$ and $p$.

Next we check the regularized 13-moment system by the two criteria
proposed at the beginning of this section.
\subsubsection{Linearized equations at the equilibrium}\label{sec:linearized}
The primary idea of Grad's moment method is that assuming the
distribution function is not far from the equilibrium, and
approximating the distribution by polynomials with the equilibrium as
the weight function. Hence, it is expected that at the equilibrium,
the regularized moment system is a high-order approximation of Grad's
moment system. So the linearized equations of the regularized moment
system and Grad's moment system must be same. For a given equilibrium
$\bw^{0}=\bw_{eq}$, assume $\bw = \bw^{0}+\epsilon \hat{\bw}$, where
$\epsilon$ is a small quantity. The linearized equations of the
quasi-linear equations
$\pd{\bw}{t}+\bA(\bw)\pd{\bw}{x}=\boldsymbol{Q}$ is
$\pd{\hat{\bw}}{t}+\bA(\bw^{0})\pd{\hat{\bw}}{x}=\frac{1}{\epsilon}\boldsymbol{Q}$.
Since at the equilibrium, $\sigma_{ij}=0$, $q_i=0$, the linearized
equations of \eqref{eq:BFGrad13-q} is
\begin{equation}\label{eq:linearized-q-Grad}
    \begin{split}
        \opd{\hat{q}_i}{t}&
        +\frac{5}{2}p^0\left(
        1-\frac{\Lihalf{3}\Lihalf{7}}{\Lihalf{5}^2}
        \right)\opd{\hat{u}_i}{t}
            +\left( \frac{7\Lihalf{9}}{2\Lihalf{7}}-\frac{5\Lihalf{7}}{2\Lihalf{5}} \right)
            \Temperature^0\pd{\hat{\sigma}_{id}}{x_d} \\
        &+\frac{5}{2}p^0
        \left( \frac{\Temperature^0}{\zz^0}\left( 1-\frac{\Lihalf{3}\Lihalf{7}}{\Lihalf{5}^2} \right)
        \pd{\hat{\zz}}{x_i}+\frac{\Lihalf{7}}{\Lihalf{5}}\pd{R\hat{T}}{x_i}
        \right)
        =\frac{1}{\epsilon}\bar{Q}_i^{(3)},
        \end{split}
        \tag{\ref{eq:BFGrad13-q}'}
\end{equation}
where $\opd{\cdot}{t}=\pd{\cdot}{t}+u^0_d\pd{\cdot}{x_d}$,
$\Li_s=-\theta\Li_s(-\theta \zz^0)$, and $\bar{Q}_i^{(3)}$ denote the
first order parts and $Q_i^{(3)}$.
The linearized equations of \eqref{eq:regularized13_framework-q} is
\begin{equation}\label{eq:linearized-q-r13}
    \begin{split}
    \opd{\hat{q}_i}{t}
    &+\frac{5}{2}p^0\left(
    1-\frac{\Lihalf{3}\Lihalf{7}}{\Lihalf{5}^2}
    \right)\opd{\hat{u}_i}{t}
    +\left( \frac{7\Lihalf{9}}{2\Lihalf{7}}-\frac{5\Lihalf{7}}{2\Lihalf{5}} \right)
    \Temperature^0\pd{\hat{\sigma}_{id}}{x_d} \\
    &+\frac{5}{2}\Temperature^0\left(
    \frac{7\frac{\Lihalf{7}}{\Lihalf{5}}-3\frac{\Lihalf{3}}{\Lihalf{1}}}
    {5-3\frac{\Lihalf{3}^2}{\Lihalf{5}\Lihalf{1}}}-\frac{\Lihalf{7}}{\Lihalf{5}} \right)
    \left( \pd{\hat{p}}{x_i}-\frac{p^0}{\rho^0}\pd{\hat{\rho}}{x_i} \right)
    =\frac{1}{\epsilon}\bar{Q}_i^{(3)}.
    \end{split}
    \tag{\ref{eq:regularized13_framework-q}'}
\end{equation}
Easy to find the linearized equations \eqref{eq:linearized-q-Grad} and
\eqref{eq:linearized-q-r13} are different. ( It is remarked that the
linearized equations of $\rho$, $u_i$, and $\sigma_{ij}$ of the two
moment system are same.) Hence, the linearized
equations of Grad's 13-moment system and the regularized-13 moment
system are different, which indicates the regularized 13-moment system
\eqref{eq:regularized13_framework} is not a good approximation of the
QBE. 

\subsubsection{NSF law}\label{sec:NSF1}
In Chapmann-Enskog expansion, the equilibrium corresponds to the
$0$-th order expansion\cite{QuantumGrad13}. For the first-order
Chapmann-Enskog expansion, the major term of the regularized 13-moment
system \eqref{eq:regularized13_framework} is expected to be same as
that of the quantum Grad's 13-moment system, which indicates the
regularization only modifies some high-order terms of the quantum
Grad's 13-moment system in the sense of Chapmann-Enskog expansion.

For the QBGK collision, the first step of Maxwellian iteration
\cite{Ikenberry} of the regularized 13-moment system
\eqref{eq:regularized13_framework} yields the Navier-Stokes-Fourier
law as
\begin{equation}
    \begin{aligned}
        \sigma_{ij}^{(1)}&=-2\tau p\pd{u_{\langle i}}{x_{j\rangle}},\\
        q_i^{(1)}&= -\tau \frac{5}{2}p\left(
        1-\frac{\Lihalf{3}\Lihalf{7}}{\Lihalf{5}^2} \right)
        \left(-\frac{1}{\rho}\right)\pd{p}{x_i}
        -\tau\frac{5}{2}
        \Temperature\left(
        \frac{7\frac{\Lihalf{7}}{\Lihalf{5}}-3\frac{\Lihalf{3}}{\Lihalf{1}}}
        {5-3\frac{\Lihalf{3}^2}{\Lihalf{5}\Lihalf{1}}}-\frac{\Lihalf{7}}{\Lihalf{5}} \right)
        \left( \pd{p}{x_i}-\frac{p}{\rho}\pd{\rho}{x_i} \right)\\
        &\neq-\frac{5}{2}\tau p\left(
        \frac{7}{2}\frac{\Lihalf{7}}{\Lihalf{5}}-\frac{5}{2}\frac{\Lihalf{5}}{\Lihalf{3}}
        \right)\pd{\Temperature}{x_i}.
    \end{aligned}
\end{equation}
Therefore, the regularized 13-moment system can not give the correct
Fourier law. In this sense, the regularized 13-moment system is also
not a proper approximation of the QBE.
}

\subsection{Regularized 13-moment system}
\add{
In Sect. \ref{sec:hyperbolicity}, we have pointed out that the quantum
Grad's 13 moment system is not hyperbolic even around the equilibrium.
And in the upper subsection, we proposed a regularized 13-moment
system based on the framework in \cite{framework, Fan2015}. However,
the moment system does not satisfy the two criteria, i.e. it fails to
share the same linearized equation of Grad's 13-moment system and can
not give the correct NSF law. In this subsection, we focus on a new
regularized moment system, which not only is hyperbolic but also
satisfies the two criteria. 

The quantum Grad's 13-moment system is essentially equivalent to
\eqref{eq:BoltzmannmP}, i.e $\mP\pd{\mP f}{t}+\mP v_d\pd{\mP f}{x_d} =
\mP Q(\mP f, \mP f)$, and it is not globally hyperbolic but satisfies
the two criteria. The regularized 13-moment system
\eqref{eq:regularized13_framework} is essentially equivalent to
\eqref{eq:BoltzmannmPP}, i.e $\mP\pd{\mP f}{t}+\mP v_d\mP\pd{\mP
f}{x_d} = \mP Q(\mP f, \mP f)$, and it is globally hyperbolic but does
not satisfy the two criteria.
As is pointed out in Sect. \ref{sec:hme13framewrok}, the equilibrium
plays an essential role in Grad's moment method, and is also essential
in the two criteria.
}
So we split the expansion \eqref{eq:Grad13expansion1} into two parts:
the equilibrium part and the non-equilibrium part
\begin{equation}
    \mP f = f_{eq} + (\mP f - f_{eq}),
\end{equation}
and use the treatment in \eqref{eq:BoltzmannmP} to deal with the
equilibrium part, and the treatment in \eqref{eq:BoltzmannmPP} to deal
with the non-equilibrium part:
\begin{equation}\label{eq:framework_BFR2}
    \mP\pd{f_{eq}}{t}+\mP\xi_d\pd{f_{eq}}{x_d} +
    \mP\pd{\mP f-f_{eq}}{t}+\mP\xi_d\mP\pd{\mP f-f_{eq}}{x_d}=\mP Q(\mP f, \mP f).
\end{equation}

The upper equation is equivalent to
\begin{equation}\label{eq:framework_BFR2-2}
    \mP\pd{\mP f}{t}+\mP\xi_d\pd{f_{eq}}{x_d} +
    \mP\xi_d\mP\pd{\mP f-f_{eq}}{x_d}=\mP Q(\mP f, \mP f).
\end{equation}
Direct calculation yields
\begin{equation}\label{eq:BF_xifeq}
    \begin{aligned}
    \mP\xi_d\pd{f_{eq}}{x_d} &= u_d\mP\pd{f_{eq}}{x_d}
        +\rho\pd{u_d}{x_d}\phi^{(0)}+\pd{p}{x_i}\phi_i^{(1)}\\
        &\qquad
        {\tiny
            +\left( 
            2p\pd{u_{\langle i}}{x_{j\rangle}}
            +\frac{2\delta_{ij}\mathfrak{b}}{3}
             p\pd{u_d}{x_d} \right)
            \phi_{ij}^{(2)}
        }
        + \frac{5p}{2}\pd{\Temperature\frac{\Lihalf{7}}{\Lihalf{5}}}{x_i}
        \phi_i^{(3)},
    \end{aligned}
\end{equation}
and
\begin{equation}\label{eq:BF_xifnoneq}
    \begin{aligned}
    \mP\xi_d\mP\pd{\mP f-f_{eq}}{x_d} &= u_d\mP\pd{\mP f-f_{eq}}{x_d}
        +\pd{\sigma_{id}}{x_d}\phi_i^{(1)}\\
        &\quad{
            +\left( 
            \frac{4}{5}\pd{q_{\langle i}}{x_{j\rangle}}
            +\frac{4}{5}\sigma_{k\langle i}\pd{u_{k}}{x_{j\rangle}}
            + \frac{2\delta_{ij}\mathfrak{b}}{3} \left( \sigma_{id}\pd{u_i}{x_d}+\pd{q_d}{x_d} \right) \right)
            \phi_{ij}^{(2)}
        }\\
        &\quad+ 
        \Temperature \left( \frac{7\Lihalf{9}}{2\Lihalf{7}}-\frac{5\Lihalf{7}}{2\Lihalf{5}} \right) \pd{\sigma_{id}}{x_d}
        \phi_i^{(3)}.
    \end{aligned}
\end{equation}
Collecting \eqref{eq:timedeverivate}, \eqref{eq:BF_xifeq},
\eqref{eq:BF_xifnoneq} and \eqref{eq:collisionterm}, and matching the
coefficients of $\phi^{(0)}$, $\phi_i^{(1)}$, $\phi_{ij}^{(2)}$ and
$\phi_i^{(3)}$, we obtain the new regularized 13-moment system:
\begin{subequations}\label{eq:regularized13_final}
    \begin{align}
        \opd{\rho}{t}&+\rho\pd{u_d}{x_d}=0,\\
        \rho\opd{u_i}{t}&+\pd{p_{id}}{x_d}=0,~i=1,2,3,\\
        \begin{split}
            \opd{\sigma_{ij}}{t}&
            +2p\pd{u_{\langle i}}{x_{j\rangle}}
            +\frac{4}{5}\pd{q_{\langle i}}{x_{j\rangle}}
            +\frac{4}{5}\sigma_{d\langle i}\pd{u_{d}}{x_{j\rangle}}\\
            &+\delta_{ij}\mathfrak{b}
            \left(  \opd{p}{t}-\frac{p}{\rho}\opd{\rho}{t} 
            + \frac{2}{3} \left( p_{kd}\pd{u_k}{x_d}+\pd{q_d}{x_d} \right)
            \right)
            =Q_{ij}^{(2)},
        \end{split}\\
        \begin{split}
            \opd{q_i}{t}&
            +\frac{5}{2}\left(
            p-\rho\Temperature\frac{\Lihalf{7}}{\Lihalf{5}} \right)\opd{u_i}{t}
            +\sigma_{ij}\opd{u_j}{t} 
            +\frac{5}{2}p\pd{\Temperature\frac{\Lihalf{7}}{\Lihalf{5}}}{x_i}\\
            &+\Temperature \left( \frac{7\Lihalf{9}}{2\Lihalf{7}}-\frac{5\Lihalf{7}}{2\Lihalf{5}} \right) \pd{\sigma_{id}}{x_d}
            =Q_i^{(3)}.
        \end{split}
    \end{align}
\end{subequations}

\add{
Similar as that for \eqref{eq:regularized13_framework}, comparing the
quantum Grad's 13-moment system \eqref{eq:BFGrad13} and
the regularized 13-moment system \eqref{eq:regularized13_final},
one can observe that only the governing equations of $\sigma_{ij}$ and
$q_i$ are different, and the two moment systems share the same
governing equations of $\rho$, $u_i$ and $p$.

Next, we check the hyperbolicity and the two criteria of the moment
system \eqref{eq:regularized13_final}.
\subsubsection{Linearized equations at the equilibrium}
\label{sec:linearized2}
Following the method in Sect. \ref{sec:linearized}, we linearize the
quantum Grad's 13-moment system \eqref{eq:BFGrad13} and the
regularized 13-moment system \eqref{eq:regularized13_final}.
Direct calculation yields that both the linearized equation of 
\eqref{eq:BFGrad13} and \eqref{eq:regularized13_final} are
\begin{subequations}
    \begin{align}
        \opd{\hat{\rho}}{t}&+\rho^0\pd{\hat{u}_d}{x_d}=0,\\
        \rho^0\opd{\hat{u}_i}{t}&+\pd{\hat{p}_{id}}{x_d}=0,~i=1,2,3,\\
        \begin{split}
            \opd{\hat{\sigma}_{ij}}{t}&
            +2p^0\pd{\hat{u}_{\langle i}}{x_{j\rangle}}
            +\frac{4}{5}\pd{\hat{q}_{\langle i}}{x_{j\rangle}}\\
            &+\delta_{ij}\mathfrak{b}
            \left(  \opd{\hat{p}}{t}-\frac{p^0}{\rho^0}\opd{\hat{\rho}}{t} 
            + \frac{2}{3} \left( p^0\pd{\hat{u}_d}{x_d}+\pd{\hat{q}_d}{x_d} \right)
            \right)
            =\frac{1}{\epsilon}\bar{Q}_{ij}^{(2)},
        \end{split}\\
        \begin{split}
            \opd{\hat{q}_i}{t}&
            +\frac{5}{2}p^0\left(
            1-\frac{\Lihalf{3}\Lihalf{7}}{\Lihalf{5}^2}
            \right)\opd{\hat{u}_i}{t}
            +\left( \frac{7\Lihalf{9}}{2\Lihalf{7}}-\frac{5\Lihalf{7}}{2\Lihalf{5}} \right)
            \Temperature^0\pd{\hat{\sigma}_{id}}{x_d} \\
            &+\frac{5}{2}p^0
            \left( \frac{\Temperature^0}{\zz^0}\left( 1-\frac{\Lihalf{3}\Lihalf{7}}{\Lihalf{5}^2} \right)
            \pd{\hat{\zz}}{x_i}+\frac{\Lihalf{7}}{\Lihalf{5}}\pd{\hat{\Temperature}}{x_i}
            \right)
                =\frac{1}{\epsilon}\bar{Q}_i^{(3)},
        \end{split}
    \end{align}
\end{subequations}
where $\opd{\cdot}{t}=\pd{\cdot}{t}+u^0_d\pd{\cdot}{x_d}$,
$\Li_s=-\theta\Li_s(-\theta \zz^0)$, and $\bar{Q}_{ij}^{(2)}$ and
$\bar{Q}_i^{(3)}$ denote the first order parts of $Q_{ij}^{(2)}$ and
$Q_i^{(3)}$, respectively.
Hence, the regularized 13-moment system \eqref{eq:regularized13_final}
satisfy the first criterion.

\subsubsection{NSF law}
Similar as that in Sect. \ref{sec:NSF1}, we derive the NSF law from the
regularized 13-moment system \eqref{eq:regularized13_final}.  For the
QBGK collision, the first step of Maxwellian iteration
\cite{Ikenberry} of the regularized 13-moment system
\eqref{eq:regularized13_final} yields the Navier-Stokes-Fourier law as
\begin{equation}
    \begin{aligned}
        \sigma_{ij}^{(1)}&=-2\tau p\pd{u_{\langle i}}{x_{j\rangle}},\\
        q_i^{(1)}&= -\tau \frac{5}{2}\left( p-\rho\Temperature\frac{\Lihalf{7}}{\Lihalf{5}} \right)
        \left(-\frac{1}{\rho}\right)\pd{p}{x_i}
        -\tau\frac{5}{2}p\pd{\Temperature\frac{\Lihalf{7}}{\Lihalf{5}}}{x_i}\\
        &=-\frac{5}{2}\tau p\left(
        \frac{7}{2}\frac{\Lihalf{7}}{\Lihalf{5}}-\frac{5}{2}\frac{\Lihalf{5}}{\Lihalf{3}}
        \right)\pd{\Temperature}{x_i}.
    \end{aligned}
\end{equation}
This is the correct NSF law, same as that of the quantum Grad's
13-moment system \cite{QuantumGrad13}.
}

\subsubsection{Hyperbolicity} \label{sec:hyperbolicityproof}
In this subsection, we study the hyperbolicity of the regularized
13-moment system\eqref{eq:regularized13_final}.  Since
$\sigma_{ij}=p_{ij}-\frac{\delta_{ij}}{3}p_{kk}$ and
$p=\frac{1}{3}p_{kk}$, denote
\[
    \bw=(\rho,u_1,u_2,u_3,p_{11},p_{12},p_{13},p_{22},p_{23},p_{33},
    q_1,q_2,q_3)^T,
\]
then the regularized 13-moment system \eqref{eq:regularized13_final}
can be written as
\begin{equation}\label{eq:DMD}
    \bD\opd{\bw}{t}+\bM_d\bD\pd{\bw}{x_d}=\boldsymbol{Q},
\end{equation}
where $\boldsymbol{Q}$ is the same as that in \eqref{eq:ms13}, and
\begin{small}
\begin{equation}
    \bD = \left( \begin{array}{ccccccccccccc}
        1&0&0&0&0&0&0&0&0&0&0&0&0 \\ 
        0&\rho&0&0&0&0&0&0&0&0&0&0&0\\ 
        0 &0&\rho&0&0&0&0&0&0&0&0&0&0\\ 
        0&0&0&\rho&0&0&0&0&0&0 &0&0&0\\ 
        -\frac{p}{\rho}\mathfrak{b} &0&0&0&
        \frac{\mathfrak{b}+2}{3}&0&0&\frac{\mathfrak{b}-1}{3}&0&\frac{\mathfrak{b}-1}{3}&0&0&0 \\ 
        0&0&0&0&0&1&0&0&0&0&0&0&0\\ 
        0&0 &0&0&0&0&1&0&0&0&0&0&0\\ 
        -\frac{p}{\rho}\mathfrak{b} &0&0&0&
        \frac{\mathfrak{b}-1}{3}&0&0&\frac{\mathfrak{b}+2}{3}&0&\frac{\mathfrak{b}-1}{3}&0&0&0 \\ 
        0&0&0&0&0&0&0&0&1&0&0&0&0 \\ 
        -\frac{p}{\rho}\mathfrak{b} &0&0&0&
        \frac{\mathfrak{b}-1}{3}&0&0&\frac{\mathfrak{b}-2}{3}&0&\frac{\mathfrak{b}+2}{3}&0&0&0 \\ 
        0 & \mathfrak{d}p+\sigma_{11} & \sigma_{12} & \sigma_{13} &
        0 & 0 & 0 & 0 & 0 & 0 & 1 & 0 & 0\\
        0 & \sigma_{12} & \mathfrak{d}p+\sigma_{22} & \sigma_{23} &
        0 & 0 & 0 & 0 & 0 & 0 & 0 & 1 & 0\\
        0 & \sigma_{13} & \sigma_{23} & \mathfrak{d}p+\sigma_{33} &
        0 & 0 & 0 & 0 & 0 & 0 & 0 & 0 & 1\\
    \end{array}
 \right),
\end{equation}
\end{small}
where 
\[
    \mathfrak{d} = \frac{5}{2}\left(
    1-\frac{\Lihalf{3}\Lihalf{7}}{\Lihalf{5}^2}\right).
\]
Matrices $\bM_d$ can be obtained from
\eqref{eq:regularized13_final}, and particularly, $\bM_1$ is
\begin{footnotesize}
    \begin{equation}\label{eq:defbM}
    \left(
    \begin {array}{ccccccccccccc} 
    0&1&0&0&0&0&0&0&0&0&0&0&0 \\ 
    \noalign{\medskip}\frac{p}{\rho}&0&0&0
    &1+m_2&0&0&m_2&0&m_2&0&0&0\\ 
    0&0&0&0&0&1 &0&0&0&0&0&0&0\\ 
    0&0&0&0&0&0&1&0&0&0&0&0&0 \\ 
    0&2m_1&0&0&0&0&0&0&0&0&\frac{2\mathfrak{b}}{3}+\frac{8}{15}&0&0\\ 
    0&0&m_1&0&0&0&0&0&0&0&0&\frac{2}{5}&0 \\ 
    0&0&0&m_1&0&0&0&0&0&0&0&0&\frac{2}{5}\\ 
    0&0&0&0&0&0&0&0&0&0&\frac{2\mathfrak{b}}{3}-\frac{4}{15}&0&0\\ 
    0&0&0&0&0&0&0&0&0&0&0&0 &0\\ 
    0&0&0&0&0&0&0&0&0&0&\frac{2\mathfrak{b}}{3}-\frac{4}{15}&0&0\\ 
    0&0&0&0&m_3+\frac{2}{3}m_1&0&0&m_3-\frac{1}{3}m_1&0&m_3-\frac{1}{3}m_1&0&0&0\\ 
    0&0&0&0&0&m_1&0&0&0&0&0&0&0 \\ 
    0&0&0&0&0&0&m_1&0&0&0&0&0&0
\end{array}
\right),
\end{equation}
\end{footnotesize}
\[
    \begin{aligned}
        m_1 &=\frac{\Lihalf{7}}{\Lihalf{5}}\Temperature,\qquad\qquad\qquad
        m_2 = \frac{1}{2}\left( 1-\frac{\Lihalf{5}^2}{\Lihalf{3}\Lihalf{7}} \right),\\
        m_3&=\frac{5}{6\mathfrak{b}}\frac{\Temperature}{5\frac{\Lihalf{1}}{\Lihalf{3}}-3\frac{\Lihalf{3}}{\Lihalf{5}}}
        \left( 2\frac{\Lihalf{1}\Lihalf{7}}{\Lihalf{3}\Lihalf{5}}
        -3\left(1-\frac{\Lihalf{3}\Lihalf{7}}{\Lihalf{5}^2}\right) \right).
    \end{aligned}
\]
If we rewrite the system \eqref{eq:regularized13_final} into the
quasi-linear form as 
\begin{equation}\label{eq:hme13final}
    \pd{\bw}{t}+\bA_d^{R}\pd{\bw}{x_d}=\boldsymbol{Q},
\end{equation}
then $\bA_d^R=\bD^{-1}\left( \bM_d+u_d\identity \right)\bD$. 
\add{
\begin{theorem}\label{thm:hyperbolic}
    The regularized 13-moment system \eqref{eq:hme13final} is globally
    hyperbolic.
\end{theorem}
We use the technique in \cite{Grad13toR13} to prove the theorem.
Firstly, we present a lemma about matrix and its eigenvalues without
proof.
\begin{lemma}\label{lem:diagonalizable}
    For a square matrix $\bA\in\bbR^{n\times n}$, denote $\lambda_i$,
    $i=1,\cdots,k$, ($k\leq n$) by the all distinct eigenvalues of
    $\bA$, and $p(\bA)=\prod_{i=1}^k(\bA-\lambda_i\identity)$, then
    $\bA$ is real diagonalizable if and only if $\lambda_i$ are all
    real and $p(\bA)=0$.
\end{lemma}
\begin{proof}[Proof of Thm. \ref{thm:hyperbolic}]
    Since both the QBE and the Hilbert space $\bbHzut$ are Galilean
    invariant, the system \eqref{eq:hme13final} is Galilean invariant.
    We just need to prove $\bA_1^R$ is real diagonalizable, which is
    equivalent to $\bM_1$ is real diagonalizable. Direct calculation
    yields the characteristic polynomial of $\bM_1$ is
    \begin{equation}\label{eq:cp_M1}
        |\lambda \identity-\bM_1|={\lambda}^5\left(
        {\lambda}^2-\frac{7\Lihalf{9}}{5\Lihalf{7}}\Temperature
        \right)^2 \left( {\lambda}^4-c_1\Temperature{\lambda}^2+
        c_0(\Temperature)^2\right),
    \end{equation}
    where $c_0$ and $c_1$ are same as that in \eqref{eq:cp_13eq}.
    Easy to check all the eigenvalues of $\bM_1$ are real and all
    zeros of $\lambda^4-c_1\Temperature\lambda^2+c0(\Temperature)^2$
    are nonzero and distinct (all the eigenvalues of $\bM_1$ will
    be studied numerically in the later of this subsection.)
    \begin{itemize}
        \item Case 1: $\frac{7\Lihalf{9}}{5\Lihalf{7}}$ is not a zero
            of $x^2-c_1x+c0$. Let
            \begin{equation}
                p(\bM_1)=\bM_1\left(
                \bM_1^2-\frac{7\Lihalf{9}}{5\Lihalf{7}}\Temperature \right)
                \left( \bM_1^4-c_1\Temperature\bM_1^2+ c_0(\Temperature)^2\right).
            \end{equation}
            With the help of Maple\footnote{Maple is a trademark of Waterloo
            Maple Inc.}, we can directly check $p(\bM_1)=0$. Using Lem.
            \ref{lem:diagonalizable}, we prove that $\bM_1$ is
            diagonalizable.
        \item Case 2: $\frac{7\Lihalf{9}}{5\Lihalf{7}}$ is a zero
            of $x^2-c_1x+c0$. 
            We note that this case only occur for Fermion with
            $\zz\approx 11.69$.
            Let
            \begin{equation}
                p(\bM_1)=\bM_1
                \left( \bM_1^4-c_1\Temperature\bM_1^2+ c_0(\Temperature)^2\right).
            \end{equation}
            With the help of Maple, we can also check $p(\bM_1)=0$. Using Lem.
            \ref{lem:diagonalizable}, we prove that $\bM_1$ is
            diagonalizable.
    \end{itemize}
    This completes the proof.
\end{proof}

In the proof, we get the characteristic polynomials of the matrix
$\bM_1$ in \eqref{eq:cp_M1}, and all the eigenvalues of the
regularized 13-moment system are given in \eqref{eq:eigenvaluesG13}.
In Figs.  \ref{fig:eigenvaluesBoson} and \ref{fig:eigenvaluesFermi},
we present the positive eigenvalues of $\bM_1$ as the fugacity varies.
For Fermion, due to Pauli exclusion principle, the high-speed
particles increase faster than the low-speed particles as the fugacity
increasing, thus the propagation speed increase. Even though, there
exists an upper bound for the propagation speed for any $\zz$.  It is
worth to remind that for $\zz\approx 11.69$, there is a point of
intersection in the right figure of Fig. \ref{fig:eigenvaluesFermi}.
Analogously, for Boson, more particles are staying on the ground
states that the propagation speeds decrease as the fugacity
decreasing.

\begin{figure}[ht]
    \centering
    \begin{overpic}[width=.5\textwidth]{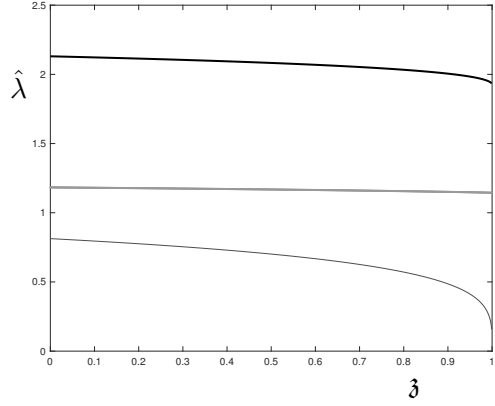}
        \put(76,1){$\zz$}
        \put(6,53){$\hat{\lambda}$}
    \end{overpic}
    \caption{\label{fig:eigenvaluesBoson} Normalized eigenvalues of
$\bM_1$ for Boson.}
\end{figure}
\begin{figure}
    \subfigure{
        \begin{overpic}[width=.5\textwidth]{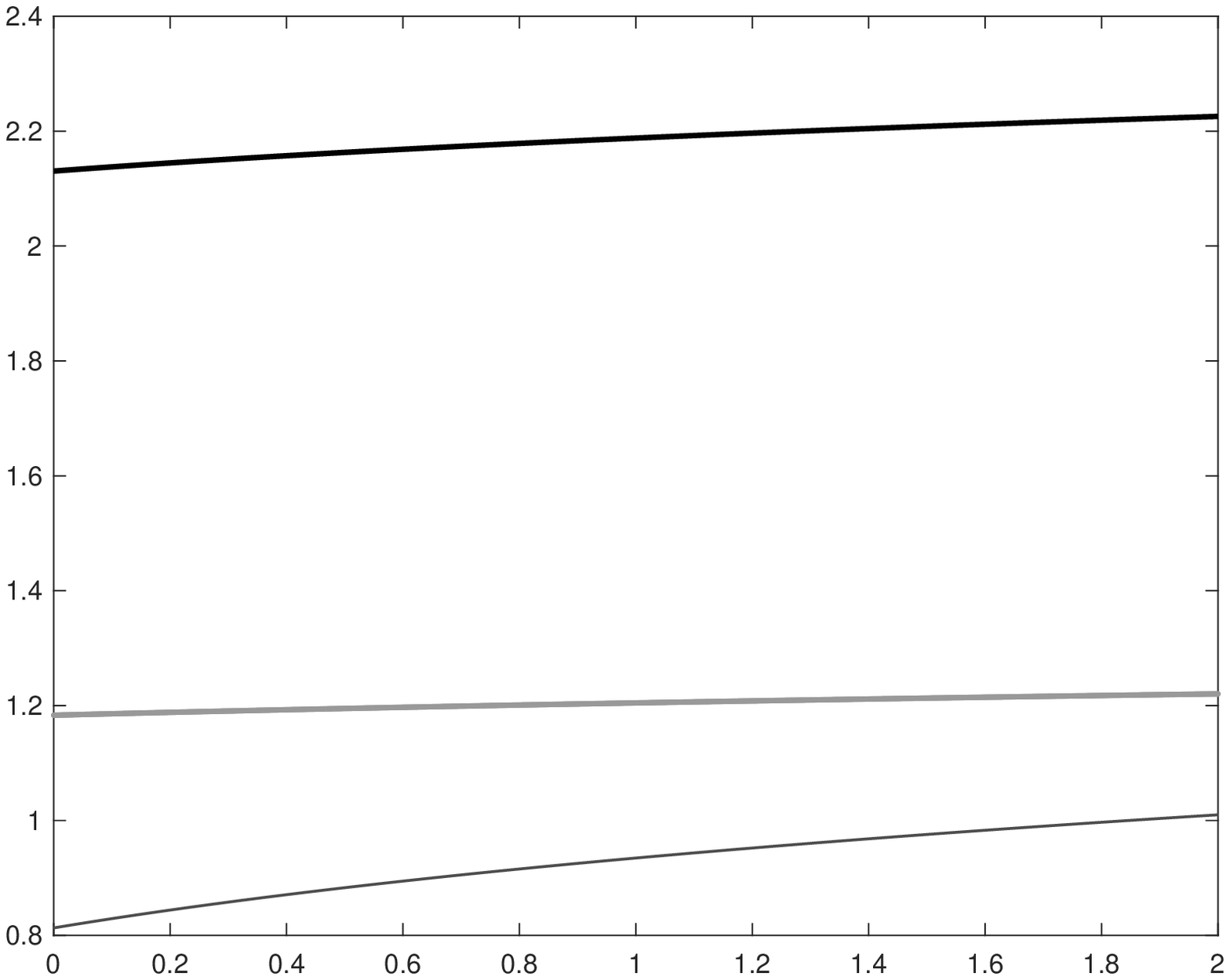}
            \put(76,1){$\zz$}
            \put(6,53){$\hat{\lambda}$}
        \end{overpic}
    }
    \subfigure{
        \begin{overpic}[width=.5\textwidth]{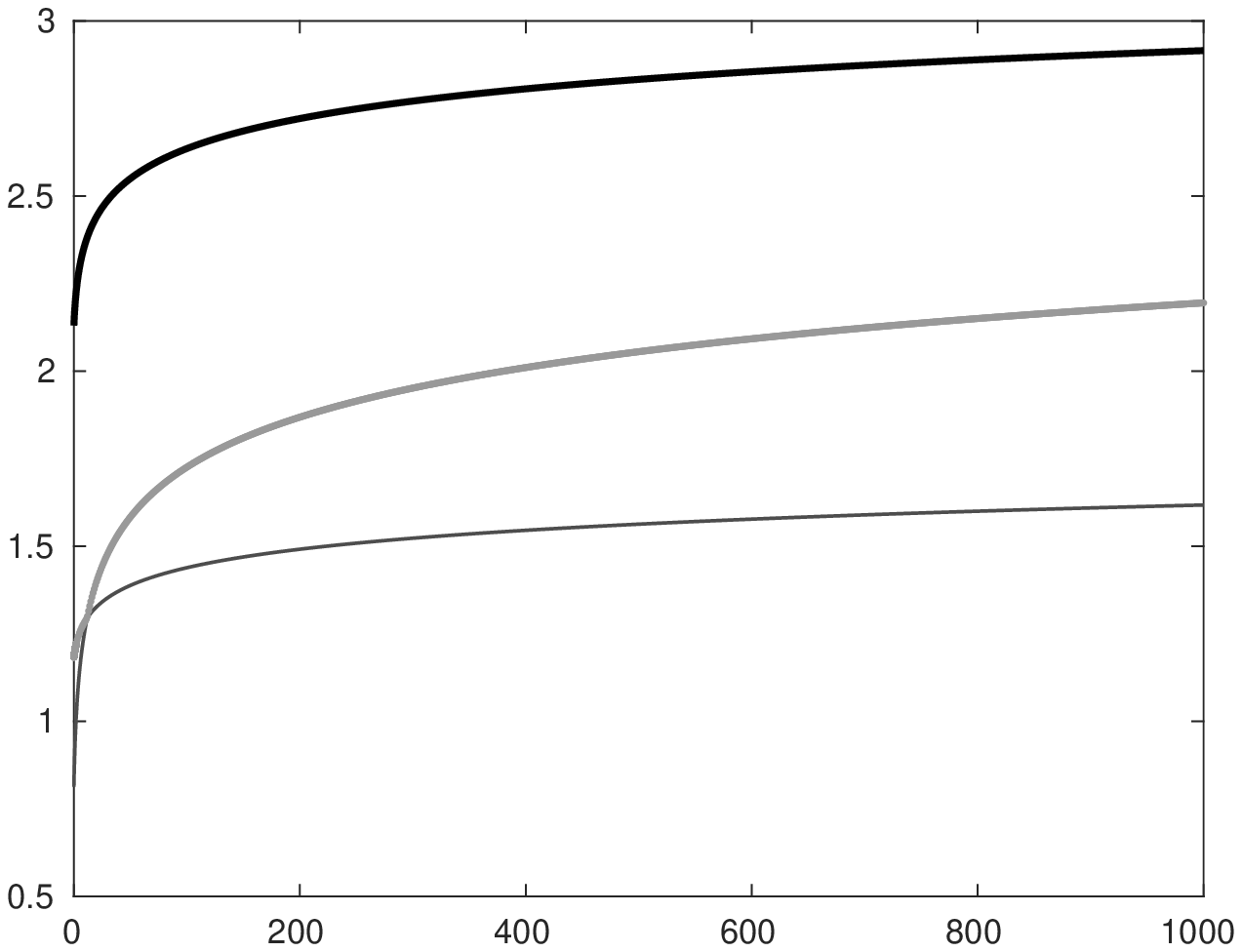}
            \put(76,1){$\zz$}
            \put(6,53){$\hat{\lambda}$}
        \end{overpic}
    }
    \caption{\label{fig:eigenvaluesFermi} Normalized eigenvalues of
$\bM_1$ for Fermion.}
\end{figure}

}

\add{
\subsubsection{Discussion}\label{sec:discussion}
We have verified the regularized 13-moment system
\eqref{eq:regularized13_final} is not only hyperbolic but also
satisfies the two criteria. Here we explore the connotation of the
model.

For the QBE, $\pd{\cdot}{t}$ is the time derivatives operator, and
$v_d\pd{\cdot}{x_d}$ is the convection operator. The convection
operator can be split into the product of the multiplying velocity
operator $v_d\cdot$ and the space derivative operator
$\pd{\cdot}{x_d}$.  For Grad's moment method \eqref{eq:BoltzmannmP},
$\mP \pd{\cdot}{t}$ and $\mP v_d\pd{\cdot}{x_d}$ are the corresponding
time derivative operator and convection operator, respectively. But
the convection operator can not be split into the product of two
operator anymore.
If we rewrite the quantum Grad's 13-moment system \eqref{eq:BFGrad13}
as
\begin{equation}
    \bD\pd{\bw}{t}+\bB_d\pd{\bw}{x_d}=\boldsymbol{Q},
\end{equation}
where $\bB_d=\bD\bA_d-u_d\identity$, then comparing it with
\eqref{eq:QBE} and \eqref{eq:BoltzmannmP}, we can find the
following corresponding relation:
\begin{equation}
    f\longrightarrow \mP f \longleftrightarrow \bw,\qquad
    \pd{\cdot}{t}\longrightarrow\mP\pd{\cdot}{t} \longleftrightarrow \bD\pd{\cdot}{t},\qquad
    v_d\pd{\cdot}{x_d}\longrightarrow\mP v_d\pd{\cdot}{x_d} \longleftrightarrow \bB_d\pd{\cdot}{x_d}.
\end{equation}
In other words, $\bD\pd{\cdot}{t}$ and $\bB_d\pd{\cdot}{x_d}$ are the
corresponding time derivative operator and convection operator,
respectively.

Analogously, for the regularized 13-moment system
\eqref{eq:regularized13_final}, comparing \eqref{eq:DMD} with
\eqref{eq:QBE}, we can find the following corresponding relation:
\begin{equation}
    f \longrightarrow \bw,\qquad
    \pd{\cdot}{s} \longrightarrow \bD\pd{\cdot}{s},\quad s=t,x_d,\qquad
    v_d\cdot \longrightarrow \bM_d,\qquad
    v_d\pd{\cdot}{x_d}\longrightarrow\bM_d\bD\pd{\cdot}{x_d}.
\end{equation}
In this case, the convection operator $\bM_d\bD\pd{\cdot}{x_d}$ can be
split in the product of the multiplying velocity operator $\bM_d$ and
the space derivative operator $\bD\pd{\cdot}{x_d}$, which is similar
as that for the QBE. In this sense, \eqref{eq:BoltzmannmPP} is a
nature approximation of the QBE. And the key point of the regularization
is to split the convection operator into the product of the
multiplying velocity operator and the space derivative operator.

Since $\bM_d$ is a limit of the multiplying velocity operator
$v_d\cdot$ on the Hilbert space $\bbHzut$, the matrix $\bM_d$ is
expected to dependent only on $\zz$, $u_d$ and $\Temperature$, i.e the
equilibrium variables. Actually, it is true due to \eqref{eq:defbM}.
That is to say, $\bM_d(\bw)=\bM_d(\bw_{eq})$, where
$\bw_{eq}=(\rho,u_1,u_2,u_3,p, 0, 0, p, 0, p,0,0,0)$ represent the
equilibrium state. 

In Sect. \ref{sec:linearized2}, we pointed out the linearized equations
of the regularized 13-moment system \eqref{eq:regularized13_final} and
Grad's 13-moment system \eqref{eq:BFGrad13} are same. Hence,
$\bA_d(\bw_{eq})=\bA^R_d(\bw_{eq})$, i.e
$\bB_d(\bw_{eq})=\bM_d(\bw_{eq})\bD(\bw_{eq})$. Let
$\bM^G_d=\bB_d\bD^{-1}$, then $\bM^G_d(\bw)\neq \bM^G_d(\bw_{eq})$,
which indicates $\bM^G_d$ is not a limit of the multiplying velocity
operator $v_d\cdot$ on the Hilbert space $\bbHzut$.
Given $\bM_d(\bw)=\bM_d(\bw_{eq})$, we have
$\bM_d(\bw)=\bM_d^G(\bw_{eq})$. It indicates the matrices $\bM_d$ can
be also calculated by $\bM_d=\bB_d(\bw_{eq})\bD^{-1}(\bw_{eq})$.

Meanwhile, since $\bA_d(\bw_{eq})=\bA^R_d(\bw_{eq})$, Thm.
\ref{thm:hyperbolic} indicates the system \eqref{eq:ms13} is
hyperbolic on the equilibrium. 

On the other hand, the hyperbolicity yields the upper bound on the
propagation speed of the regularized moment system while the QBE
allows infinity propagation speed. Actually, deterministic methods for
the QBE, for example, the discrete velocity method and the spectral
method, also has the upper bound on the propagation speed. How to
choose the upper bound is an important issue for the deterministic
methods. If the upper bound is greater enough, the deterministic methods
work well. For the regularized 13-moment system, the upper bound is
fixed, so the application of the regularized system is limited.
However, on the other hand, we can increase the upper bound by using
more moments. In \cite{Fan_new}, the authors proposed an arbitrary
order regularized moment method, the upper bound of the propagation
speed of which can be as greater enough as the number of the moments
increasing. The regularization proposed in this subsection can be also
extended to derive arbitrary order regularized moment method. Hence,
the upper bound of the propagation speed can be enlarged by using more
moments if necessary.

}


\section{Conclusion} \label{sec:conclusion}

We study the hyperbolicity of the quantum Grad's 13-moment system. It
is found that the model is not hyperbolic even around the
thermodynamic equilibrium. Applying the moment model reduction
framework in \cite{Fan2015} in a trivial way to the quantum
Boltzmann equation, we can not obtain a good model, which is
hyperbolic and satisfies the two criteria. 
By further studying the framework in \cite{Fan2015} and the quantum
Grad's 13-moment system, we propose a regularization for the model and
obtain a globally hyperbolic 13-moment system. We are expecting that
the regularized moment system is helpful to understand the quantum
effects and to develop the moment method in quantum kinetic theory.


\section*{Acknowledge}

Y. Di was supported by the National Natural Science Foundation of
China (Grant No. 11271358). Fan and Li were supported in part by the
National Natural Science Foundation of China (Grant No. 91330205,
11421110001, 11421101 and 11325102).


\add{
\section*{Appendix}
\appendix
\section{Coefficients of \eqref{eq:cp_13eq} and \eqref{eq:def_g}}\label{app:coe}
The coefficients in \eqref{eq:cp_13eq} and \eqref{eq:def_g} can be
directly obtained by calculating the determinant of the matrix
$\bA_1$. Here we list the coefficients in \eqref{eq:cp_13eq} and
\eqref{eq:def_g} as following.
\begin{equation*}
    \begin{aligned}
        c_0 &= \frac{3\left( 7\Lihalf{3}\Lihalf{7}-5\Lihalf{5}^2
    \right)}{5\Lihalf{1}\Lihalf{5}-3\Lihalf{3}^2},\\
    　 c_1 &=\frac{140\Lihalf{1}\Lihalf{5}\Lihalf{9}
    +175\Lihalf{1}\Lihalf{7}^2 - 84\Lihalf{3}^2\Lihalf{9}
    -75\Lihalf{3}\Lihalf{5}\Lihalf{7}}
    {15\Lihalf{7}\left( 5\Lihalf{1}\Lihalf{5}-3\Lihalf{3}^2 \right)},\\
       c_2 &= 
       \frac{1}{\Lihalf{3}^2\Lihalf{7}^3\left(
       5\Lihalf{1}\Lihalf{5}-3\Lihalf{3}^2 \right)}
       \Bigg[ -735{\Lihalf{3}}^{3}{\Lihalf{7}}^{3}\Lihalf{9}
        +560\,{\epsilon}^{2}\Lihalf{1}{\Lihalf{5}}^{2}{\Lihalf{7}}^{2}
        \left( \Lihalf{5}\Lihalf{ {9}}-{\Lihalf{7}}^{2} \right)\\
        &\qquad\qquad+294 {\Lihalf{3}}^{2}{\Lihalf{5}}^{2}\Lihalf{9} 
        \left( {\epsilon}^{2}\Lihalf{5}\Lihalf{9} -
        \left( \frac{17\epsilon^2}{7}-\frac{25}{14} \right)\Lihalf{7}^2
         \right)\\
        &\qquad\qquad+{\epsilon}^{2}\Lihalf{3}{\Lihalf{5}}^{2}\Lihalf{7}\left(
        196\Lihalf{1}{\Lihalf{9}}^{2}
        - {210{\Lihalf{5}}^{2}\Lihalf{9}}
        +450\Lihalf{5}{\Lihalf{7}}^{2} \right)
    \Bigg], \\
    c_3 &= 
    \frac{1}{3\Lihalf{7}^2\Lihalf{3}^2\left( 5\Lihalf{1}\Lihalf{5}-3\Lihalf{3}^2 \right)}
    \Bigg[ -588{\Lihalf{3}}^{4}{\Lihalf{9}}^{2}
    +720 \,{\epsilon}^{2}\Lihalf{1}{\Lihalf{5}}^{3}{\Lihalf{7}}^{2}\\
        &\qquad\qquad+ {\Lihalf{3}}^{3}\left( -525
    \,\Lihalf{5}\Lihalf{7}\Lihalf{9}+1575\,{\Lihalf{7} }^{3} \right) \\
    &\qquad\qquad +  {\Lihalf{3}}^{2}\left(  \left(  \left( -432\,{ \epsilon}^{2}-1125
    \right) {\Lihalf{5}}^{2}+1225\,\Lihalf{1} \Lihalf{9} \right)
    {\Lihalf{7}}^{2}+980\,\Lihalf{1}\Lihalf{5}{\Lihalf{9}}^{2} \right)
    \Bigg],\\
    c_4&=\frac{\left( -1225\,\Lihalf{1}\Lihalf{9}+375\,\Lihalf{3}\Lihalf{7}
    \right)
    \Lihalf{5}-875\,\Lihalf{1}{\Lihalf{7}}^{2}+735\,{\Lihalf{3}}^{2}\Lihalf{9}}
    {3\Lihalf{7}\left(5\Lihalf{1}\Lihalf{5}-3\Lihalf{3}^2\right)}.
\end{aligned}
\end{equation*}

}


\bibliographystyle{plain}
\bibliography{../article,../tiao}
\end{document}